\documentclass[a4paper]{article}
\usepackage[margin=25mm]{geometry}
\usepackage{graphicx}
\usepackage{subfigure}
\usepackage{float}
\usepackage{graphicx}
\usepackage{amssymb}
\usepackage{mathrsfs}
\usepackage{fullpage}
\usepackage{amsmath}
\usepackage{amsthm}
\usepackage{graphicx,caption}
\usepackage{eucal}
\usepackage{lipsum}
\usepackage{amsfonts}
\DeclareMathAlphabet{\mathpzc}{OT1}{pzc}{m}{it}
\usepackage{epsfig}
\usepackage{adjustbox}
\allowdisplaybreaks
\newtheorem{theorem}{Theorem}[section]
\usepackage{latexsym}
\usepackage[nodisplayskipstretch]{setspace}
\makeatletter
\def\ps@pprintTitle{%
\let\@oddhead\@empty
\let\@evenhead\@empty
\def\oddfoot\@empty
\let\evenfoot\@oddfoot}
\makeatother
\usepackage{verbatim}
\usepackage{lineno}

\providecommand{\keywords}[1]
{
  \small	
  \textbf{\textit{Keywords---}} #1
}
\title {Optimization of Traffic Control in $\textrm{\it MMAP[k]/PH[k]/S}$ Catastrophic Queueing Model with $\textit{P\!H}$ Retrial Times and Preemptive Repeat  Policy 
} 

\author{Raina Raj$^{1}$, Vidyottama Jain$^{1}$   \thanks{Corresponding author}   \\
        \small Central University of Rajasthan, Ajmer, India$^{1}$\\
    }
  
\date{} 

\begin{document}
\maketitle

\begin{abstract}
The presented study elaborates a multi-server catastrophic retrial queueing model considering  preemptive repeat priority policy with phase-type ($\textrm{\it P\!H}$) distributed retrial times. 
For the sake of comprehension, the scenario of model operation prior and later to the occurrence of the disaster is referred to as the normal scenario and as the catastrophic scenario, respectively.
In both scenarios, the arrival and service processes of all types of calls follow marked Markoian arrival process ($\textrm{\it M\!M\!A\!P}$) and $\textrm{\it P\!H}$ distribution with distinct parameters, respectively. 
 In the normal scenario, the incoming heterogeneous calls are categorized as handoff calls and new calls.  An arriving new call will be blocked when all the channels are occupied, and consequently, will join the orbit (virtual space) of infinite capacity. From the orbit, the blocked new call can either retry for the service or exit the system following $\textrm{\it P\!H}$ distribution. Whereas, an arriving handoff call is given preemptive repeat priority over a new call in service when all the channels are occupied and at least one of the channel is occupied with a new call otherwise the handoff call is dropped, and consequently, this preempted new call will join the orbit. 
In the catastrophic scenario, when a disaster causes the shut down of the entire system and failure of all functioning channels,  a set of backup channels is quickly deployed to restore services. These failed channels are repaired following $\textrm{\it P\!H}$ distribution. The bursty nature of the catastrophe arrival is represented by Markovian arrival process ($\textit{M\!A\!P}$). Due to the emergency situation in the concerned area, calls to or from the emergency personnel are classified as emergency calls. Therefore, the incoming heterogeneous calls are divided into three categories: handoff calls, new calls, and emergency calls. Emergency calls are provided absolute preemptive priority over handoff calls and new calls due to the pressing need to save lives  in such situations.
Behaviour of the proposed system is modelled by the level-dependent-quasi birth death ($\textrm{\it L\!D\!Q\!B\!D}$) process. The Markov chain's ergodicity criteria are established by demonstrating that it belongs to the class of asymptotically quasi-Toeplitz Markov chains ($\textrm{\it A\!Q\!T\!M\!C}$). For the approximate computation of the stationary distribution, a new approach is developed.
 The expressions of various performance measures have been derived for the numerical illustration. An optimization problem to obtain optimal value of total number of  backup channels  has been formulated and dealt by employing non dominated sorting genetic algorithm-II (NSGA-II) approach.
\end{abstract} \hspace{10pt}


\section{Introduction} \label{section1}

Over the past decade,  there has been discovered a substantial attention  in queuing models with disastrous events in communications and cellular networks. These disastrous events are either human-induced (e.g., fire, virus attack, power outages, etc.) or natural (e.g., flood, tsunami, cyclone, etc.), and are typically mentioned as catastrophe events. With the occurrence of a catastrophic event,  all active and waiting customers/traffic are compelled to exit the system abruptly, rendering it inoperable. For a comprehensive survey of catastrophic events occurring in  communication networks, refer to \cite{dabrowski2015catastrophic} and references cited therein. For these types of catastrophic queueing models, prioritization of traffic is also important  in many cellular network applications  e.g., in multiprocessor switching, voice traffic is provided priority over data traffic. The system's various aspects, such as arrival discipline, service discipline, categories of services, and so on, ensure traffic prioritization. Many priority policies, such as guard channel policy, threshold policy, preemptive priority policy, non-preemptive priority policy, and others, have been proposed in the literature to  this end (refer, \cite{brandwajn2017multi,chang1965preemptive,krishnamoorthy2008map,machihara1995bridge}). 
In this investigation,  preemptive repeat priority policy has been implemented for the higher priority traffic. When all of the channels are occupied, but at least one of these channels is occupied by a lower priority traffic, the arriving higher priority traffic will preempt the service of the ongoing lower priority traffic. This preempted lower priority traffic will then join a virtual space, called orbit, and it will retry for its service from the scratch. The process of recurrent attempts to obtain the service is referred as retrial phenomenon and  it's been extensively explored by the researchers in the area of communications and cellular networks \cite{kim2016survey}.


 An overview of research works on the catastrophic queueing models without retrial phenomenon can be found in the articles (refer, \cite{yajima2019central,baumann2012steady,yechiali2007queues}). Sudhesh et al. \cite{sudhesh2017transient} considered a   two heterogeneous servers queueing model with catastrophic event, server failure and repair.   In their work, the waiting customers were presumed to be impatient and may exit the system at any time, when the system was down. However, the applicability of  the above mentioned models has been diminished in the present scenario, since the incoming call arrival followed Poisson process and service times was considered to be exponentially distributed. In contrast to the memory-less property of stationary Poisson flow, the input stream of arrival contains burstiness and correlation properties.    Thus, more generalized arrival and services processes are employed, such as Markovian arrival process ($\textit{M\!A\!P}$), marked Markovian\ arrival process ($\textit{M\!M\!A\!P}$) and phase-type ($\textit{P\!H}$) distribution.
   Some of the  pertinent studies with the consideration of more general processes  are as follows. 
     Chakravarthy \cite{chakravarthy2017catastrophic} presented a $\textit{M\!A\!P/P\!H/1}$ queueing model considering catastrophic event and delayed action. Recently, Kumar and Gupta \cite{kumar2020analysis} studied a discrete–time catastrophic model with population arrival following batch Bernouli process, binomial catastrophe arrival occurred according to discrete-time renewal process according to which each individual  either survived with probability $p$ or died with probability $1-p$.


  Reviews of some of the relevant literature for the catastrophic retrial queueing models are as follows.
  Wang et al. \cite{wang2008transient}  proposed a $M/G/1$ retrial queueing model with catastrophe phenomenon. In this study, the inter-retrial time was exponentially distributed and catastrophe arrival occurred according to Poisson process. On the similar track, Chakravarthy et al. \cite{chakravarthy2010retrial}  presented a $\textrm{\it M\!A\!P/P\!H/1}$ retrial queueing model with catastrophe phenomenon and repair process. The arrival of catastrophe followed Poisson process and failed channels were repaired following exponential distribution. Their model was a confined one due to the assumption of exponential distribution for retrial and repair processes. Recently, Ammar and Rajadurai \cite{ammar2019performance} proposed a $M/G/1$ preemptive priority retrial queue with  working breakdown services and disasters. They assumed  when a failure occurred in the system, the major server was sent for the repair and the replacement server worked at a slow rate of service. Though the above mentioned studies had considered generalized arrival and/or service processes, yet retrial process was described through exponential distribution only. In wireless cellular networks, the inter-retrial times are notably brief in comparison to the service times. Since, the retrial attempt is just a matter of pushing one button, these retrial customers will make numerous attempts during any given service interval. Therefore, the consideration of exponential retrial times in place of non-exponential ones could lead to under or over estimating the system parameters as shown by various studies in the literature (refer, \cite{chakravarthy2020retrial,dharmaraja2008phase,jain2021,shin2011approximation}). Raj and Jain \cite{raj2021} presented a $\textrm{\it M\!M\!A\!P[2]/P\!H[2]/S}$ queueing model with $\textrm{\it P\!H}$ distributed retrial times and preemptive repeat priority policy. The proposed a traffic control optimization problem and employed heuristic approaches to obtain the optimal solution. Therefore, to obtain realistic performance measures for retrial phenomenon,  a more generalized approach,  $\textrm{\it P\!H}$ distribution  has been applied in the provided model.

In this work,  a  multi-server catastrophic  queueing model with  $\textrm{\it P\!H}$ distribution for retrial process and preemptive repeat priority policy is introduced.   To the best of authors' knowledge, the proposed model is the first one that deals with such complex system. 
For the sake of clarity, the scenario of model functioning prior to the occurrence of  a disaster (man-made/natural) is referred to as the normal scenario, and after the disaster, as the catastrophic scenario. In the normal scenario, the system will  provide services to all incoming calls; however, in the catastrophic scenario, the system will collapse, flushing out all active and waiting calls. $\textrm{\it M\!M\!A\!P}$ and $\textrm{\it P\!H}$ distributions with varying parameters are used to describe the arrival and service processes of all types of incoming calls, respectively.  
In normal scenario, the incoming calls are classified as handoff call and new call. The new call, which finds all the channels busy upon its arrival will join the orbit (virtual space) of infinite capacity and will be referred as a retrial call (\cite{jain2020numerical}). Following $\textit{P\!H}$ distribution, the retrial call can either retry for service or quit the system without receiving it. When all the channels are occupied and at least one of the channels is occupied with the new call, an arriving handoff call is given preemptive priority over it; otherwise, the arriving handoff call is lost from the system. The preempted new call will join the orbit and retry for its service from the scratch.

 In the catastrophic scenario, when a calamity strikes, the normally operating system is entirely shut down. Due to the sudden outbreak, all active and retrial calls are forced to terminate their processes and are removed from the system. Since a catastrophic event is considered bursty in nature,  $\textrm{\it M\!A\!P}$  is an apt representation of the disaster's arrival phase. Following $\textrm{\it P\!H}$ distribution, the failed channels  are  repaired promptly. When all of the channels fail in the system, a set of backup/standby channels is instantly installed in the affected area, and services are immediately restored. Calls to and from emergency services, such as hospitals, police, and fire departments, should be given precedence over other public calls in such tragic circumstances. Thus, the incoming heterogeneous calls are now classified as handoff calls, new calls and emergency calls.  When all channels are occupied by  either handoff calls or new calls, an arriving emergency call is given absolute preemptive priority over either handoff calls or new calls. If all the channels are occupied by both handoff and new calls, an arriving emergency call will preempt the service of an ongoing new call. The arriving emergency call will start receiving service in place of the preempted handoff or new call. These preempted calls will be lost from the system as there is no concept of orbit in this scenario.  The underlying process of the presented system is modelled by level dependent quasi-birth-death $\textit{(L\!D\!Q\!B\!D)}$ process.
The detailed study over   $\textit{L\!D\!Q\!B\!D}$ process can be found in \cite{he2014fundamentals} and \cite{latouche1999introduction}.
Ergodicity conditions of the underlying Markov chain are obtained by proving that the Markov chain satisfies the properties of asymptotically quasi-Toeplitz Markov chains ($\textrm{\it A\!Q\!T\!M\!C}$) \cite{klimenok2006multi}. A new algorithm is developed for efficient computation of the steady-state distribution. Further, the expressions of various performance measures have been derived for the numerical illustration. Due to the  consideration of the preemptive  priority policy, the blocking probability for emergency calls decreases and simultaneously the frequent termination of services for handoff and new calls increases the probability of preemption. Thus, an optimization problem to obtain optimal value of total number of  backup channels  has been formulated and dealt by employing non-dominated sorted genetic algorithm-II (NSGA-II) approach \cite{deb2002fast}.

In this work, a  multi-dimensional Markov chain has been constructed by appropriately  choosing the components for $\textrm{\it P\!H}$ distributed  service and retrial times. There are two different ways to keep track of the phases for service and retrial processes. One of them, referred to as the TPFS (track-phase-for-server) in \cite{he2018space}, keeps track of the current phase of service at each busy server. The second approach, known as the CSFP (count-server-for-phase) in \cite{he2018space} , counts the number of servers having a certain phase of service. The TPFS approach yields a simpler and more transparent form of the blocks of the generator of the multi-dimensional Markov chain. However, if the $\textrm{\it P\!H}$ distribution has $N$ phases with $n$ number of busy channels, then there are $Nn$ states of the service process. CSFP approach makes the multi-dimensional Markov chain less transparent. However, the number of states of the process is equal to $T_n = (n+N-1)!/n! (N-1)!.$ Thus, for $N=2$ and $n=20,$ the number of states are 1048576 for TPFS approach and 21 for the CSFP approach. 
This approach justifies  to use the CSFP method in this work which is helpful in the efficient computation of steady-state distribution.

The layout of this work is  arranged in seven sections.  In Section \ref{section2},  a  $\textrm{\it M\!M\!A\!P[k]/P\!H[k]/S}$ model with  $\textrm{\it P\!H}$ distributed retrial times  and catastrophe phenomenon is described.  In Section \ref{section3},  the infinitesimal generator matrix for the proposed $\textit{L\!D\!Q\!B\!D}$ process has been derived. The ergodicity condition of the underlying process is obtained  by proving that the Markov chain belongs to the class of $\textrm{\it A\!Q\!T\!M\!C}$. An algorithm is proposed to compute steady-state probabilities.  In Section \ref{section4}, formulas of key performance measures to analyse the system efficiency  are derived explicitly. Numerical  illustrations to point out the impact of various intensities over the system performance are presented in Section \ref{section5}. An optimization problem has been formulated to evaluate the  behaviour of the system in Section \ref{section6}. Finally, the underlying model is concluded with the insight for the future works in Section \ref{section7}.

\section{Model Description} \label{section2}
\subsection{Details and Assumptions}

The presented study introduce with a $\textrm{\it M\!M\!A\!P[k]/P\!H[k]/S}$ catastrophic model with $\textrm{\it P\!H}$ distributed retrial times. Here, $S$ is defined as the total number of channels in the system. The proposed approach can be observed to work in two  scenarios: normal and catastrophic. The model's operation prior to the onset of a disaster can be described as normal, and the latter as catastrophic.

\begin{itemize}
    \item[-]\textbf{ Normal Scenario:}   
     In this scenario,  the incoming heterogeneous calls are categorized in two classes ($k=2$) as handoff calls and new calls. Both types of calls have distinct arrival and service processes that follow the $\textit{M\!M\!A\!P}$ and $\textit{P\!H}$ distributions with different parameters, respectively. If all of the channels are occupied when a new call arrives, the new call will join the orbit (virtual space) of infinite capacity and be referred to as a retrial call (\cite{jain2020numerical}).
     The retrial call following $\textit{P\!H}$ distribution can either retry for service or exit the system without obtaining the service. It's been considered that when the number of retrial calls $\mathpzc{l}$ is between 0 and $M$, the retrial rate is $\theta_{\mathpzc{l}}$ and once the number exceed $M,$ the retrial rate is considered as $\theta.$ One of the following two scenarios may occur when a handoff call arrives to the system and all available channels are occupied. The arriving handoff call will be lost from the system if all of the channels are occupied by handoff calls in the first scenario. In the second scenario, among the occupied channels, if at least one of the channels is occupied with a new call, the arriving handoff call will be given preemptive priority over the new call  in progress. The handoff call commenced its service in place of the preempted new call and this preempted new call  joins the orbit.  
    
\item[-] \textbf{ Catastrophic Scenario:}    The presented model is subject to catastrophic events such as power outages, virus assaults, natural disasters, fires, and so on.  With the occurrence of such disaster events, all the calls in the system (the one in service and the one waiting for service) leave the system prematurely  and  the system becomes inactivated.  Since a catastrophic event is bursty by nature, $\textit{M\!A\!P}$ is an appropriate representation of the disaster's arrival phase. The failed channels  are immediately  repaired following $\textit{P\!H}$ distribution.
In this study, it is considered that when the whole system is collapsed and all channels are failed, $K$ backup/standby channels will start providing services at slow rate. These backup channels will stop working once one of the failed channels is fixed. Arriving calls are deemed lost from this point onward until all the channels are fixed.
The occurrence of a disaster causes an emergency situation in the affected area. Thus,  the incoming heterogeneous calls are now categorized in three classes ($k=3$) as handoff calls, new calls and emergency calls. The arrival and service processes of all types of calls follow $\textit{M\!M\!A\!P}$ and  $\textit{P\!H}$ distributions with distinct parameters, respectively.  To provide emergency services in catastrophic scenario, a general sense of priority should be attached to the  emergency calls. Here, it is assumed that the emergency calls are provided  preemptive priority over the handoff and new calls. When an emergency call arrives to the system, out of the following three cases, one might occur.
 \begin{itemize}
     \item[1.] When an arriving emergency call finds  all the channels  occupied with the emergency calls only,  the arriving emergency call will be lost form the system.
     \item[2.] If at the arrival epoch of an emergency call, all the channels are occupied with handoff  calls only, the service of a handoff call will be preempted and the emergency call will start receiving service in place of that preempted handoff call.
     \item[3.] If at the arrival epoch of an emergency call,  all the channels are occupied and at least one of the channels is occupied with a new call, the emergency call will preempt the service of that new call and commence service in its place.
 \end{itemize}
 
    It's been considered that there is no orbit for the blocked or preempted calls when backup channels are working. As, the calls in the orbit need to wait some time before receiving service and it is unreasonable to make calls wait significant time before being admitted when there are urgent needs to save lives or properties.

\end{itemize}

All the other assumptions are described in Table \ref{tab:my_label2}.

\begin{table}[]
	\centering
	\scalebox{0.8}
	{
	\begin{tabular}{|l|l|}
	\hline
	$S$ & total number of channels\\
	\hline
	$K$ & total number of backup channels\\
	\hline
	$C_0, C_{\mathcal{N}}, C_{\mathcal{H}}, C_{\mathcal{E}}$ & the square matrices of size $L_1$that characterize the \textit{M\!M\!A\!P}\\
	\hline
	$\lambda_{\mathcal{H}}, \lambda_{\mathcal{N}}, \lambda_{\mathcal{E}}$ &  the average arrival rates of handoff, new and emergency calls\\
	\hline
	$(\beta_{\mathcal{H}}, A_{\mathcal{H}})$ & representation of  \textit{$P\!H$} distribution of handoff call with dimension $M_{\mathcal{H}}$\\
	\hline
		$(\beta_{\mathcal{N}}, A_{\mathcal{N}})$ & representation of  \textit{$P\!H$} distribution of new call with dimension $M_{\mathcal{N}}$\\
	\hline
		$(\beta_{\mathcal{E}}, A_{\mathcal{E}})$ & representation of  \textit{$P\!H$} distribution of emergency call with dimension $M_{\mathcal{E}}$\\
	\hline
	$\mu_{\mathcal{H}},\mu_{\mathcal{N}},\mu_{\mathcal{E}}$ & the average service rates of handoff, new and emergency calls\\
	\hline
		$(\gamma, \Gamma)$ & representation of  \textit{$P\!H$} distribution of retrial call with dimension $N$\\
			\hline
			$\Gamma^0(1)$,$\Gamma^0(2)$ &  the absorption due to departure from the cell and   the absorption due to retrial attempt\\
				\hline
				$ \theta$ & the average retrial rate of retrial call\\
				\hline
			$D_0, D_{1}$ & the square matrices of size $L_2$ that characterize the \textit{M\!A\!P}\\
			\hline
			$(\alpha, B)$ & representation of  \textit{$P\!H$} distribution of repair process with dimension $R$\\
			\hline
			$T_n^m$ & $(m+n-1)!/(m-1)!n!$\\
			\hline
			$\Tilde{A} $ & $\begin{pmatrix}
			0 & 0\\
			A^0 & A
			\end{pmatrix}$\\
			\hline
			$P_{\kappa_1}(\beta_{\mathcal{H}})$ &  the matrix that defines the transition probabilities of the process at the epoch of starting new process given \\ & that $\kappa_1$  channels are busy \cite{klimenok2006multi}\\
			\hline
			$L_{S-\kappa_1-\kappa_2}(S-\kappa_2,\Tilde{A_H})$ &  the matrix that defines the transition intensities of the process  at the service completion epoch given that \\ & $\kappa_1$ handoff calls are busy at this epoch \cite{klimenok2006multi}\\
			\hline
			$A_{\kappa_1}(S-\kappa_2,A_H)$ & the matrix that defines the transition intensities of the process which do not lead to the service completion \\ & given that  $\kappa_1$ handoff calls are busy \cite{klimenok2006multi}\\
			\hline
			$\bigotimes$ and $\bigoplus$ &  the Kronecker product and sum of matrices, respectively, see \cite{dayar2012analyzing} \\
			\hline
			diag & main diagonal of a matrix\\
			\hline 
			$\text{diag}^{+}$ & upper diagonal of a matrix\\
			\hline
			$\text{diag}^{-}$ & lower diagonal of a matrix\\
			\hline
			row & row vector\\
			\hline
			col & column vector\\
			\hline
	\end{tabular}}
	\caption{Notation}
	\label{tab:my_label2}
\end{table}

\section{Mathematical Analysis} \label{section3}


The underlying process \{$\Xi(t), t \geq 0 \}$ for a cell is defined by the following state space:
\begin{align*}
	\nonumber \Omega &= \{(\mathpzc{l}, \kappa_1, \kappa_2,  \mathfrak{j}, i, v_1, v_2, s_{\mathcal{H}}^1, s_{\mathcal{H}}^2,\ldots,s_{\mathcal{H}}^{M_\mathcal{H}}, s_{\mathcal{N}}^1, s_{\mathcal{N}}^2,\ldots,s_{\mathcal{N}}^{M_\mathcal{N}}, s_{\mathcal{E}}^1, s_{\mathcal{E}}^2,\ldots,s_{\mathcal{E}}^{M_\mathcal{E}}, r^1, r^2,\ldots, r^{N}, v^1, v^2,\ldots, v^{R});\\ &~~~~~~ \mathpzc{l} \geq 0,~0 \leq \kappa_1 \leq S,~0 \leq \kappa_2 \leq S,~ 0 \leq \mathfrak{j} \leq S,~0 \leq i \leq K,~1 \leq v_1 \leq L_1,~1 \leq v_2 \leq L_2\},\end{align*}
where, 
\begin{itemize}
	\item $\mathpzc{l}$ is the number of retrial calls,
	\item $\kappa_1$ is the number of handoff calls in the system receiving service,
	\item $\kappa_2$ is the number of new calls in the system receiving service,
	\item $\mathfrak{j}$ is the number of channels under repair,
	\item $i$ is the number of emergency calls in the system receiving service,
	\item $v_1$ is the current phase of $\textrm{\it M\!M\!A\!P}$ for arrival of calls,
		\item $v_2$ is the current phase  of $\textrm{\it M\!A\!P}$ for arrival of catastrophe,
	\item $s_{\mathcal{H}}^{m_1}$ be the number of channels for handoff calls which are in phase $m_1$ ; $s_{\mathcal{H}}^{m_1} = \overline{0,S},$ $m_1 = \overline{1,M_\mathcal{H}}$,
	\item $s_{\mathcal{N}}^{m_2}$ be the number of channels for new calls which are in phase $m_2$ ; $s_{\mathcal{N}}^{m_1} = \overline{0,S},$ $m_2 = \overline{1,M_\mathcal{N}}$,
	\item $s_{\mathcal{E}}^{m_3}$ be the number of channels for emergency calls which are in phase $m_3$ ; $s_{\mathcal{E}}^{m_3} = \overline{0,K},$ $m_3 = \overline{1,M_\mathcal{E}}$,
	\item $v^{h_2}$ be the number of channels under repair which are in phase $h_2$ ; $v^{h_2} = \overline{0,S},$ $h_2 = \overline{1,R}$,
	\item $r^{h_1}$ be the number of retrial calls  which are in phase $h_1$ ; $r^{h_1} \geq 0,$ $h_1 = \overline{1,N}$.
	\end{itemize}

	The stochastic process \{$\Xi(t), t \geq 0 \}$ can be modelled as level-dependent quasi birth death  ($\textrm{\it L\!D\!Q\!B\!D}$) process with the  infinitesimal generator matrix provided as follows:
\begin{center}
	$\mathscr{Q} =
	\begin{pmatrix}
	\mathscr{Q}^{0} & \mathscr{Q}_{0,1} & 0 & 0 & 0 & 0 &  \\
	\mathscr{Q}^{'}& \mathscr{Q}_{1,1}& \mathscr{Q}_{1,2} & 0 & 0 & 0 &  \\
	\mathscr{Q}_{2,0}  & \mathscr{Q}_{2,1} & \mathscr{Q}_{2,2} & \mathscr{Q}_{2,3} & 0 &  0 &  \\
	\mathscr{Q}_{3,0}  & 0 &\mathscr{Q}_{3,2} & \mathscr{Q}_{3,3} & \mathscr{Q}_{3,4} &   0 &  \\
	\vdots & \vdots  & \vdots & &\ddots & \ddots & \ddots \\
	\vdots& \vdots  & \vdots & & & \ddots & \ddots & \ddots \\
		\mathscr{Q}_{M,0}  & 0 &0 &&& \mathscr{Q}_{M,M-1} &\mathscr{Q}_{M,M} & \Tilde{Q} &  &   \\
\mathscr{Q}^{+}  & 0 &0 &&&& \mathscr{Q}_{0} & \mathscr{Q}_{1} & \mathscr{Q}_{2}    &   \\
\vdots  & \vdots & \vdots &&&& &\ddots & \ddots & \ddots       \\

	\end{pmatrix}.$\\
\end{center}

{\small{
\begin{align*}
& \textbf {Upper  Diagonal :}\\
& \mathscr{Q}_{\mathpzc{l},\mathpzc{l}+1} = \text{diag}\{ X_\mathpzc{l}(0),X_\mathpzc{l}(1), \ldots, X_{\mathpzc{l}}(S)\} + \text{diag}^+\{ \hat{X}_\mathpzc{l}(0), \hat{X}_\mathpzc{l}(1),\ldots,\hat{X}_{\mathpzc{l}}(S-1)\};\mathpzc{l} \geq 0,\\
	&  X_\mathpzc{l}(\kappa_1) = \textrm{diag}\{X_\mathpzc{l}(\kappa_1,\kappa_2)\};\kappa_1 = \overline{0,S}, \kappa_2 = \overline{0,S-\kappa_1},\\
  & X_0(\kappa_1,\kappa_2) = 
  \begin{cases}
col(X_0(0,0,j));\mathfrak{j} = \overline{0,S} ,\\
      col(X_0(\kappa_1,\kappa_2,0),X_0(\kappa_1,\kappa_2,S)); \\ \forall S=K \text{or}~ K<S ~~\&~~ \kappa_1+\kappa_2 \leq K,\\
      X_0(\kappa_1,\kappa_2,0);  \text{$\forall K<S \& \kappa_1+\kappa_2 > K$,}
  \end{cases} X_0(\kappa_1,\kappa_2,\mathfrak{j}) = \begin{cases}
X_0(\kappa_1,\kappa_2,\mathfrak{j},0); 
 \text{$\forall \mathfrak{j} = \overline{0,S-1}$},\\
col(X_0(\kappa_1,\kappa_2,S,i));\forall i=\overline{0,K-\kappa_1-\kappa_2},  \end{cases} \\
      & \hat{X}_\mathpzc{l}(\kappa_1)  =
  \begin{pmatrix}
      \hat{X}_\mathpzc{l}(\kappa_1,0)&  & \\
      \hat{X}_\mathpzc{l}(\kappa_1,1)&  & \\
      & \ddots & \\
       & & \\
      & & \hat{X}_\mathpzc{l}(\kappa_1,S-\kappa_1)
  \end{pmatrix}; \kappa_1 = \overline{0,S-1},
  \hat{X}_0(\kappa_1,\kappa_2) = 
  \begin{cases}
  col(\hat{X}_0(0,0,\mathfrak{j})); \forall \mathfrak{j}=\overline{0,S},\\
      col(\hat{X}_0(\kappa_1,\kappa_2,0),\hat{X}_0(\kappa_1,\kappa_2,S)); \\ \forall S=K  ~\text{or}~~ K<S  ~~\&~~ \kappa_1+\kappa_2 \leq K,\\
      \hat{X}_0(\kappa_1,\kappa_2,0);  \text{$\forall K<S \& \kappa_1+\kappa_2 > K$,}
  \end{cases}\\& \hat{X}_0(\kappa_1,\kappa_2,\mathfrak{j}) = \begin{cases}
  \hat{X}_0(\kappa_1,\kappa_2,\mathfrak{j},0); 
 \text{$\forall \mathfrak{j} = \overline{0,S-1}$},\\
col(\hat{X}_0(\kappa_1,\kappa_2,S,i)); \forall  i = \overline{0,K-\kappa_1-\kappa_2},   
  \end{cases} \\
&  X_\mathpzc{l}(\kappa_1,\kappa_2) = X_\mathpzc{l}(\kappa_1,\kappa_2,0)= X_\mathpzc{l}(\kappa_1,\kappa_2,0,0);\forall \mathpzc{l} \geq 1,  ~~\hat{X}_\mathpzc{l}(\kappa_1,\kappa_2) = \hat{X}_\mathpzc{l}(\kappa_1,\kappa_2,0)= \hat{X}_\mathpzc{l}(\kappa_1,\kappa_2,0,0);\forall \mathpzc{l} \geq 1,\\ 
& X_\mathpzc{l}(\kappa_1,\kappa_2,\mathfrak{j},i) =  C_{\mathcal{N}} \otimes I_{L_2T^{M_{\mathcal{H}}}_{\kappa_1}T^{M_{\mathcal{N}}}_{\kappa_2}} \otimes P_\mathpzc{l}(\gamma);  \mathpzc{l} \geq 0,\kappa_1  = \overline{0,S}, \kappa_2 = S-\kappa_1, \mathfrak{j}=i=0, \\
& \hat{X}_\mathpzc{l}(\kappa_1,\kappa_2,\mathfrak{j},i) =   C_{\mathcal{H}} \otimes I_{L_2} \otimes P_{\kappa_1}(\beta_{\mathcal{H}}) \otimes I_{(T^{M_{\mathcal{N}}}_{\kappa_2} \times T^{M_{\mathcal{N}}}_{\kappa_2-1})} \otimes P_\mathpzc{l}(\gamma); \mathpzc{l} \geq 0, \kappa_1  = \overline{0,S-1}, \kappa_2 = S-\kappa_1, \mathfrak{j}=i=0. \\
&  \Tilde{Q} = \mathscr{Q}_{M,M+1}, \Tilde{Q}= \mathscr{Q}_{2},\\
& X_{M}(\kappa_1,\kappa_2,\mathfrak{j},i) =  C_{\mathcal{N}} \otimes I_{L_2T^{M_{\mathcal{H}}}_{\kappa_1}T^{M_{\mathcal{N}}}_{\kappa_2}} \otimes P_{M}^{'}(\gamma);\kappa_1  = \overline{0,S}, \kappa_2 = S-\kappa_1, \mathfrak{j}=i=0, \\
& \hat{X}_{M}(\kappa_1,\kappa_2,\mathfrak{j},i) =   C_{\mathcal{H}} \otimes I_{L_2} \otimes P_{\kappa_1}(\beta_{\mathcal{H}}) \otimes I_{(T^{M_{\mathcal{N}}}_{\kappa_2} \times T^{M_{\mathcal{N}}}_{\kappa_2-1})} \otimes P_{M}^{'}(\gamma);  \kappa_1  = \overline{0,S-1}, \kappa_2 = S-\kappa_1, \mathfrak{j}=i=0. \\
& \text{where}~P_{M}^{'}(\gamma) ~\text{is}~ T^{N}_{M} ~\text{order square matrix.}\\
   & \textbf {Lower  Diagonal :}\\
& \mathscr{Q}^{'} =  \mathscr{Q}_{\mathpzc{l},\mathpzc{l}-1} + \mathscr{Q}^{'}_{1,0}, ~~\mathscr{Q}_{\mathpzc{l},\mathpzc{l}-1} = \text{diag}\{Z_{\mathpzc{l}}(0), Z_{\mathpzc{l}}(1), \ldots, Z_{\mathpzc{l}}(S)\}; \mathpzc{l} \geq 1,\\
&Z_{\mathpzc{l}}(\kappa_1) = \text{diag}\{ Z_{\mathpzc{l}}(\kappa_1,0),Z_{\mathpzc{l}}(\kappa_1,1), \ldots, Z_{\mathpzc{l}}(\kappa_1,S-\kappa_1)\}  + \text{diag}^+\{ \hat{Z}_{\mathpzc{l}}(\kappa_1,0), \hat{Z}_{\mathpzc{l}}(\kappa_1,1),\ldots,\hat{Z}_{\mathpzc{l}}(\kappa_1,S-\kappa_1-1)\}; \kappa_1 = \overline{0,S},\\
& Z_1(\kappa_1,\kappa_2) = 
  \begin{cases}
    row(Z_1(0,0,\mathfrak{j}));\forall \mathfrak{j}=\overline{0,S},\\
    row(Z_1(\kappa_1,\kappa_2,0),Z_1(\kappa_1,\kappa_2,S));\\  \forall S=K  \text{or} ~K<S ~\&~ \kappa_1+\kappa_2 \leq K,\\
      Z_1(\kappa_1,\kappa_2,0);  \text{$\forall K<S ~\&~ \kappa_1+\kappa_2> K$,}
  \end{cases} Z_1(\kappa_1,\kappa_2,\mathfrak{j}) = \begin{cases}
Z_1(\kappa_1,\kappa_2,\mathfrak{j},0); 
 \text{$\forall \mathfrak{j} = \overline{0,S-1}$},\\
row(Z_1(\kappa_1,\kappa_2,S,i));i=\overline{0,K-\kappa_1-\kappa_2},
  \end{cases} \\
&   \hat{Z}_1(\kappa_1,\kappa_2) = 
  \begin{cases}
      row(\hat{Z}_1(\kappa_1,\kappa_2,0),\hat{Z}_1(\kappa_1,\kappa_2,S)); \\ \forall S=K \text{or}~K<S ~\&~ \kappa_1+\kappa_2 < K,\\
      \hat{Z}_1(\kappa_1,\kappa_2,0);  \text{$\forall K<S ~\&~ \kappa_1+\kappa_2 \geq K$,}
  \end{cases} \hat{Z}_1(\kappa_1,\kappa_2,\mathfrak{j}) = \begin{cases}
\hat{Z}_1(\kappa_1,\kappa_2,\mathfrak{j},0); 
\text{$\forall \mathfrak{j} = \overline{0,S-1}$},\\
row(\hat{Z}_1(\kappa_1,\kappa_2,S,i)); i=\overline{0,K-\kappa_1-\kappa_2},   
  \end{cases} \\
  &  Z_\mathpzc{l}(\kappa_1,\kappa_2) = Z_\mathpzc{l}(\kappa_1,\kappa_2,0)= Z_\mathpzc{l}(\kappa_1,\kappa_2,0,0);\mathpzc{l} \geq 2,~~ \hat{Z}_\mathpzc{l}(\kappa_1,\kappa_2) = \hat{Z}_\mathpzc{l}(\kappa_1,\kappa_2,0)= \hat{Z}_\mathpzc{l}(\kappa_1,\kappa_2,0,0);\mathpzc{l} \geq 2,\\
& Z_\mathpzc{l}(\kappa_1,\kappa_2,\mathfrak{j},i) = 
  I_{L_1L_2T_{\kappa_1}^{M_{\mathcal{H}}}T_{\kappa_2}^{M_{\mathcal{N}}}}\otimes L_\mathpzc{l}^{(1)}(\mathpzc{l},\Tilde{\Gamma_1}); \forall \mathpzc{l} \geq 1,\kappa_1  = \overline{0,S}, \kappa_2 = \overline{0,S-\kappa_1}, \mathfrak{j}=i=0, \\
& \hat{Z}_\mathpzc{l}(\kappa_1,\kappa_2,\mathfrak{j},i) =   I_{L_1L_2T_{\kappa_1}^{M_{\mathcal{H}}}}\otimes P_{\kappa_2}(\beta_{\mathcal{N}}) \otimes L_\mathpzc{l}^{(2)}(\mathpzc{l},\Tilde{\Gamma_2}); \forall \mathpzc{l} \geq 1, \kappa_1  = \overline{0,S-1}, \kappa_2 = \overline{0,S-\kappa_1-1}, \mathfrak{j}=i=0. \\
&  \mathscr{Q}_{2} = \mathscr{Q}_{M,M-1}, \\
& Z_M(\kappa_1,\kappa_2,\mathfrak{j},i) = 
  I_{L_1L_2T_{\kappa_1}^{M_{\mathcal{H}}}T_{\kappa_2}^{M_{\mathcal{N}}}}\otimes L_M^{(1)'}(M,\Tilde{\Gamma_1}); \forall \kappa_1  = \overline{0,S}, \kappa_2 = \overline{0,S-\kappa_1}, \mathfrak{j}=i=0, \\
& \hat{Z}_M(\kappa_1,\kappa_2,\mathfrak{j},i) =   I_{L_1L_2T_{\kappa_1}^{M_{\mathcal{H}}}}\otimes P_{\kappa_2}(\beta_{\mathcal{N}}) \otimes L_M^{(2)'}(M,\Tilde{\Gamma_2}); \forall  \kappa_1  = \overline{0,S-1}, \kappa_2 = \overline{0,S-\kappa_1-1}, \mathfrak{j}=i=0. \\
& \text{where}~L_M^{(1)'}(M,\Tilde{\Gamma_1}),~L_M^{(2)'}(M,\Tilde{\Gamma_2}) ~\text{are}~ T^{N}_{M} ~\text{order square matrices.}\\
 & \textbf {Main  Diagonal :}\\
&\mathscr{Q}_{\mathpzc{l},\mathpzc{l}} = \text{diag}\{ Y_{\mathpzc{l}}(0),Y_{\mathpzc{l}}(1), \ldots, Y_{\mathpzc{l}}(S)\} + \text{diag}^+\{ \hat{Y}_{\mathpzc{l}}(0), \hat{Y}_{\mathpzc{l}}(1),\ldots,\hat{Y}_{\mathpzc{l}}(S-1)\} + \text{diag}^-\{ \Bar{Y}_{\mathpzc{l}}(1), \Bar{Y}_{\mathpzc{l}}(2),\ldots,\Bar{Y}_{\mathpzc{l}}(S)\}; \mathpzc{l} \geq 0,\\
&Y_{\mathpzc{l}}(\kappa_1) = \text{diag}\{ Y_{\mathpzc{l}}(\kappa_1,0),Y_{\mathpzc{l}}(\kappa_1,1), \ldots, Y_{\mathpzc{l}}(\kappa_1,S-\kappa_1)\} + \text{diag}^+\{ N_{\mathpzc{l}}(\kappa_1,0), N_{\mathpzc{l}}(\kappa_1,1),\ldots,N_{\mathpzc{l}}(\kappa_1,S-\kappa_1-1)\}\\
&~~~~~~~~~~~~+ \text{diag}^-\{ S^N_{\mathpzc{l}}(\kappa_1,1), S^N_{\mathpzc{l}}(\kappa_1,2),\ldots,S^N_{\mathpzc{l}}(\kappa_1,S-\kappa_1)\}; \kappa_1 = \overline{0,S},\\
 & Y_0(\kappa_1,\kappa_2) = 
  \begin{cases}
    \text{diag}\{ Y_0(0,0,0),Y_0(0,0,1), \ldots, Y_0(0,0,S)\} 
    + \text{diag}^-\{ R_0(0,0,1), R_0(0,0,2),\ldots,R_0(0,0,S)\},\\
      \text{diag}\{ Y_0(\kappa_1,\kappa_2,0),Y_0(\kappa_1,\kappa_2,S)\};  \text{$\forall S=K$ or $K<S ~\&~ \kappa_1+\kappa_2 \leq K$,}\\
      Y_0(\kappa_1,\kappa_2,0); \text{$\forall K<S ~\&~ \kappa_1+\kappa_2 > K$,}
  \end{cases}\\
  	& Y_0(\kappa_1,\kappa_2,\mathfrak{j}) = \begin{cases}
Y_0(\kappa_1,\kappa_2,\mathfrak{j},0); 
 \text{$\forall \mathfrak{j} = \overline{0,S-1}$},\\
\text{diag}\{ Y_0(\kappa_1,\kappa_2,S,0), \ldots, Y_0(\kappa_1,\kappa_2,S,K-\kappa_1-\kappa_2)\} + \text{diag}^-\{ S^E_0(\kappa_1,\kappa_2,S,1),\ldots, S^E_0(\kappa_1,\kappa_2,S,K-\kappa_1-\kappa_2)\}, \\
    + \text{diag}^+\{ E_0(\kappa_1,\kappa_2,S,0), \ldots, E_0(\kappa_1,\kappa_2,S,K-\kappa_1-\kappa_2-1)\}
    \\
   \end{cases} \\
  & N_0(0,0) =
  \begin{pmatrix}
      N_0(0,0,0) & N_0(0,0,1) &  \cdots & N_0(0,0,S-1) & 0\\
       0&0  &  \cdots & 0 & N_0(0,0,S)
  \end{pmatrix}^T,\\
  	 & N_0(\kappa_1,\kappa_2) = 
  \begin{cases}
  \text{diag}\{ N_0(\kappa_1,\kappa_2,0),N_0(\kappa_1,\kappa_2,S)\};  \text{$\forall S=K$ or $K<S ~\&~ \kappa_1+\kappa_2 < K$,} \\
       col( N_0(\kappa_1,\kappa_2,0) , 0)
  ; \text{$\forall  K<S~\&~ \kappa_1+\kappa_2 = K$,}\\
  N_0(\kappa_1,\kappa_2,0);\text{$\forall  K<S~\&~ \kappa_1+\kappa_2 > K$,} 
  \end{cases}\\
  	& N_0(\kappa_1,\kappa_2,\mathfrak{j}) = \begin{cases}
N_0(\kappa_1,\kappa_2,\mathfrak{j},0); 
 \text{$\forall \mathfrak{j} = \overline{0,S-1}$},\\
\begin{pmatrix}
      N_0(\kappa_1,\kappa_2,S,0) &  & \\
       &  \ddots & \\
      & & N_0(\kappa_1,\kappa_2,S,K-\kappa_1-\kappa_2-1)\\
      &&0
  \end{pmatrix}, 
  \end{cases} \\
   & S^N_0(\kappa_1,\kappa_2) = 
  \begin{cases}
  \begin{pmatrix}
      S^N_0(0,1,0) & S^N_0(0,1,1) &  \cdots& S^N_0(0,1,S-1)&0  \\
      0 &0 &   \cdots &0 &S^N_0(0,1,S)
  \end{pmatrix},\\
      \text{diag}\{ S^N_0(\kappa_1,\kappa_2,0),S^N_0(\kappa_1,\kappa_2,S)\};  \text{$\forall S=K$ or $K<S ~\&~\kappa_1+\kappa_2 \leq K$,}\\
       row(S^N_0(\kappa_1,\kappa_2,0),0);  \text{$\forall K<S~\&~ \kappa_1+\kappa_2 = K+1$,}\\
      S^N_0(\kappa_1,\kappa_2,0);  \text{$\forall K<S~\&~ \kappa_1+\kappa_2 > K+1$,}
  \end{cases}\\
  	& S^N_0(\kappa_1,\kappa_2,\mathfrak{j}) = \begin{cases}
S^N_0(\kappa_1,\kappa_2,\mathfrak{j},0); 
 \text{$\forall \mathfrak{j} = \overline{0,S-1}$},\\
\begin{pmatrix}
      S^N_0(\kappa_1,\kappa_2,S,0) &  & &\\
       &\ddots &  & \\
      & & S^N_0(\kappa_1,\kappa_2,S,K-\kappa_1-\kappa_2)&  S^N_0(\kappa_1,\kappa_2,S,K-\kappa_1-\kappa_2+1)
  \end{pmatrix}
  \end{cases} \\
  &  Y_\mathpzc{l}(\kappa_1,\kappa_2) = Y_\mathpzc{l}(\kappa_1,\kappa_2,0)= Y_\mathpzc{l}(\kappa_1,\kappa_2,0,0);\mathpzc{l} \geq 1,
     N_\mathpzc{l}(\kappa_1,\kappa_2) = N_\mathpzc{l}(\kappa_1,\kappa_2,0)= N_\mathpzc{l}(\kappa_1,\kappa_2,0,0);\mathpzc{l} \geq 1,\\
   &  S^N_\mathpzc{l}(\kappa_1,\kappa_2) = S^N_\mathpzc{l}(\kappa_1,\kappa_2,0)= S^N_\mathpzc{l}(\kappa_1,\kappa_2,0,0);\mathpzc{l} \geq 1,\\
 &  N_\mathpzc{l}(\kappa_1,\kappa_2,\mathfrak{j},i) =
 \begin{cases}
   C_{\mathcal{N}} \otimes I_{L_2T_{\kappa_1}^{M_{\mathcal{H}}}} \otimes P_{\kappa_2}(\beta_{\mathcal{N}}) \otimes  I_{ T_{\mathpzc{l}}^{N}};  \kappa_1 = \overline{0,S-1}, \kappa_2 = \overline{0,S-\kappa_1-1},  \mathfrak{j} =i=0, \mathpzc{l} \geq 0,\\
C_{\mathcal{N}} \otimes I_{L_2T_{\kappa_1}^{M_{\mathcal{H}}}} \otimes P_{\kappa_2}(\beta_{\mathcal{N}}) \otimes  I_{T^{R}_{\mathfrak{j}}T^{\mathcal{E}}_{i} T_{\mathpzc{l}}^{N}};  \kappa_1 = \overline{0,S-1}, \kappa_2 = \overline{0,S-\kappa_1-1},  \mathfrak{j} =S, i = \overline{0,K-\kappa_1-\kappa_2-1}, \mathpzc{l} \geq 0,\\
\end{cases}\\
 & S^N_\mathpzc{l}(\kappa_1,\kappa_2,\mathfrak{j},i) =
 \begin{cases}
 I_{L_1L_2T_{\kappa_1}^{M_{\mathcal{H}}}} \otimes L_{S-(\kappa_1+\kappa_2)}(S-\kappa_1, \Tilde{A_{\mathcal{N}}}) \otimes I_{T_{\mathpzc{l}}^{N}};  \kappa_1 = \overline{0,S-1}, \kappa_2 = \overline{1,S-\kappa_1},  \mathfrak{j} =  i =0, \mathpzc{l} \geq 0,\\  
  I_{L_1L_2T_{\kappa_1}^{M_{\mathcal{H}}}} \otimes L_{K-(\kappa_1+\kappa_2+i)}(K-\kappa_1-i, \Tilde{A_{\mathcal{N}}}) \otimes I_{T^{\mathcal{E}}_{i}T^{R}_{S} };  \kappa_1 = \overline{0,S-1}, \kappa_2 = \overline{1,S-\kappa_1}, \mathfrak{j} = S, i = \overline{0,K-\kappa_1-\kappa_2},\\
  C_{\mathcal{E}} \otimes I_{L_2T_{\kappa_1}^{M_{\mathcal{H}}}(T_{\kappa_2}^{M_{\mathcal{N}}}\times T_{\kappa_2-1}^{M_{\mathcal{N}}})} \otimes P_i(\beta_\mathcal{E}) \otimes I_{T_{S}^{R}};\kappa_1 = \overline{0,S-1}, \kappa_2 = \overline{1,S-\kappa_1}, \mathfrak{j} = S, i = K-\kappa_1-\kappa_2,
    \end{cases}\\
   &E_0(\kappa_1,\kappa_2,S,i) = C_{\mathcal{E}} \otimes I_{L_2T_{\kappa_1}^{M_{\mathcal{H}}}T_{\kappa_2}^{M_{\mathcal{N}}}} \otimes P_{i}(\beta_{\mathcal{E}})\otimes  I_{T^{R}_{\mathfrak{j}}};  \kappa_1 = \overline{0,S}, \kappa_2 = \overline{0,S-\kappa_1},  i = \overline{0,K-\kappa_1-\kappa_2-1},\\
   & S^E_0(\kappa_1,\kappa_2,S,i) = I_{L_1L_2T_{\kappa_1}^{M_{\mathcal{H}}}T_{\kappa_1}^{M_{\mathcal{N}}}} \otimes L_{K-(\kappa_1+\kappa_2+i)}(K-\kappa_1-\kappa_2, \Tilde{A_{\mathcal{E}}}) \otimes I_{T^{R}_S};  \kappa_1 = \overline{0,S}, \kappa_2 = \overline{0,S-\kappa_1},  i = \overline{1,K-\kappa_1-\kappa_2},\\
   & R_0(0,0,\mathfrak{j})=R_0(0,0,\mathfrak{j},0)=  I_{L_1L_2} \otimes L_{S-\mathfrak{j}}(S, \Tilde{B});  \mathfrak{j} = \overline{1,S},\\
   & Y_\mathpzc{l}(\kappa_1,\kappa_2,\mathfrak{j},i) = 
 \begin{cases}
  & C_0 \oplus D(1)+ \Delta;\kappa_1=\kappa_2=\mathfrak{j}=i=\mathpzc{l}=0,\\ &  (C_0+C_{\mathcal{N}}+C_{\mathcal{H}}) \oplus D(1) \oplus A_{\mathfrak{j}}(S,B) + \Delta;\kappa_1=\kappa_2=i=\mathpzc{l}=0,\mathfrak{j}= \overline{1,S-1},\\
   & C_0 \oplus D(1)\oplus A_{\kappa_1}(S-\kappa_2,A_{\mathcal{H}}) \oplus A_{\kappa_2}(S-\kappa_1,A_{\mathcal{N}}) \oplus A_{i}(K-\kappa_1-\kappa_2,A_{\mathcal{E}})\oplus A_{\mathfrak{j}}(S,B)+ \Delta; \\ & \kappa_1=\overline{0,S}, \kappa_2 = \overline{0,S-\kappa_1} \mathpzc{l}=0, \mathfrak{j}=S,i= \overline{1,K-\kappa_1-\kappa_2},\\
 & (C_0+C_{\mathcal{N}}+C_{\mathcal{H}}+C_{\mathcal{E}}) \oplus D(1)  \oplus  A_{i}(K-\kappa_1-\kappa_2,A_{\mathcal{E}})\oplus A_{\mathfrak{j}}(S,B)+ \Delta;  \\ & \kappa_1=0, \kappa_2 =0, \mathpzc{l}=0, \mathfrak{j}=S,i= K-\kappa_1-\kappa_2,\\
 & (C_0+C_{\mathcal{N}}+C_{\mathcal{H}}) \oplus D(1) \oplus A_{\kappa_1}(K-\kappa_2-i,A_{\mathcal{H}}) \oplus A_{\kappa_2}(K-\kappa_1-i,A_{\mathcal{N}}) \oplus A_{i}(K-\kappa_1-\kappa_2,A_{\mathcal{E}}) \\ & \oplus A_{\mathfrak{j}}(S,B)+ \Delta;   \kappa_1=\overline{0,S}, \kappa_2 = \overline{0,S-\kappa_1}, \mathpzc{l}=0, \mathfrak{j}=S,i= \overline{1,K-\kappa_1-\kappa_2},\\
& C_0 \oplus D_0 \oplus A_{\kappa_1}(S-\kappa_2,A_{\mathcal{H}}) \oplus A_{\kappa_2}(S-\kappa_1,A_{\mathcal{N}}) \oplus  A_{\mathpzc{l}}(\mathpzc{l},\Gamma)+ \Delta;\\
& \kappa_1=\overline{0,S}, \kappa_2 = \overline{0,S-\kappa_1-1}, \mathpzc{l} \geq 0, \mathfrak{j}=i= 0,\\
& (C_0+C_{\mathcal{H}}) \oplus D_0 \oplus A_{\kappa_1}(S-\kappa_2,A_{\mathcal{H}}) \oplus A_{\kappa_2}(S-\kappa_1,A_{\mathcal{N}}) \oplus  A_{\mathpzc{l}}(\mathpzc{l},\Gamma)\\ & +  I_{L_2T_{\kappa_1}^{M_{\mathcal{H}}}T_{\kappa_2}^{M_{\mathcal{N}}}} \otimes L_{\mathpzc{l}}^{(2)}(\mathpzc{l},\Tilde{\Gamma_2})P_{\mathpzc{l}}(\gamma)+ \Delta^{'};   \kappa_1=\overline{0,S}, \kappa_2 = S-\kappa_1, \mathpzc{l} \geq 0, \mathfrak{j}=i= 0,
 \end{cases}\\
 & \Delta = \text{diag}\{I_{L_1L_2} \otimes \Delta^{(\kappa_1, \kappa_2, \mathfrak{j}, i)} \}, \Delta^{'} = \text{diag}\{I_{L_1L_2} \otimes \Delta^{(\kappa_1, \kappa_2, \mathfrak{j}, i)^{'}} \},  \Delta^{(\kappa_1, \kappa_2, \mathfrak{j}, i)^{'}} = \Delta^{(\kappa_1, \kappa_2, \mathfrak{j}, i)} -\text{diag}\{ [I_{L_2T_{\kappa_1}^{M_{\mathcal{H}}}T_{\kappa_2}^{M_{\mathcal{N}}}} \otimes L_{\mathpzc{l}}^{(2)}(\mathpzc{l},\Tilde{\Gamma_2})P_{\mathpzc{l}}(\gamma)]e\}, \\ &\Delta^{(\kappa_1, \kappa_2, \mathfrak{j}, i)} = -\text{diag}\{ [A_{\kappa_1}(S-\kappa_2,A_{\mathcal{H}}) \oplus A_{\kappa_2}(S-\kappa_1,A_{\mathcal{N}}) \oplus  A_{i}(K-\kappa_1-\kappa_2,A_{\mathcal{E}}) \oplus A_{\mathfrak{j}}(S,B) \oplus A_{\mathpzc{l}}(\mathpzc{l},\Gamma)]e\}, \\
    & \hat{Y}_{\mathpzc{l}}(\kappa_1)=
   \begin{pmatrix}
      \hat{Y}_{\mathpzc{l}}(\kappa_1,0) &  & \\
    & \hat{Y}_{\mathpzc{l}}(\kappa_1,1) & \\
       &  &  \\
       &\ddots&\\
       &&\\
      & & \hat{Y}_{\mathpzc{l}}(\kappa_1,S-\kappa_1-1)\\
      &&0
  \end{pmatrix},
   \hat{Y}_0(\kappa_1,\kappa_2) = 
  \begin{cases}
 \begin{pmatrix}
      \hat{Y}_0(0,0,0)  &\cdots & \hat{Y}_0(0,0,S-1)&0\\
      0  &   \cdots &0   &  \hat{Y}_0(0,0,S)
  \end{pmatrix}^T,\\
      \text{diag}\{ \hat{Y}_0(\kappa_1,\kappa_2,0),\hat{Y}_0(\kappa_1,\kappa_2,S)\};\\  \forall S=K  ~\text{or}~ K<S~\&~ \kappa_1+\kappa_2 < K,\\
     col(\hat{Y}_0(\kappa_1,\kappa_2,0),\hat{Y}_0(\kappa_1,\kappa_2,S)); \\ \forall S=K  ~\text{or} ~~K<S~\&~ \kappa_1+\kappa_2 = K,\\
     \hat{Y}_0(\kappa_1,\kappa_2,0);\\  \forall S=K ~ \text{or}~ K<S~\&~ \kappa_1+\kappa_2 > K,
     \end{cases}\\
   	& \hat{Y}_0(\kappa_1,\kappa_2,\mathfrak{j}) = \begin{cases}
\hat{Y}_0(\kappa_1,\kappa_2,\mathfrak{j},0); 
 \text{$\forall \mathfrak{j} = \overline{0,S-1}$},\\
\begin{pmatrix}
      \hat{Y}_0(\kappa_1,\kappa_2,S,0) &  & \\
       &  \ddots&  \\
      & & \hat{Y}_0(\kappa_1,\kappa_2,S,K-\kappa_1-\kappa_2-1)\\
      &&0
  \end{pmatrix},
  \end{cases} \\
   &  \hat{Y}_\mathpzc{l}(\kappa_1,\kappa_2) = \hat{Y}_\mathpzc{l}(\kappa_1,\kappa_2,0)= \hat{Y}_\mathpzc{l}(\kappa_1,\kappa_2,0,0);\mathpzc{l} \geq 1,\\
  & \hat{Y}_\mathpzc{l}(\kappa_1,\kappa_2,\mathfrak{j},i) =
  \begin{cases}
    C_{\mathcal{H}} \otimes I_{L_2} \otimes P_{\kappa_1}(\beta_{\mathcal{H}}) \otimes I_{T_{\kappa_2}^{M_{\mathcal{N}}} T_{\mathpzc{l}}^{N}};  \kappa_1 = \overline{0,S-1}, \kappa_2 = \overline{0,S-\kappa_1-1}, \mathfrak{j} = i=0, \mathpzc{l} \geq 1,\\
    C_{\mathcal{H}} \otimes I_{L_2} \otimes P_{\kappa_1}(\beta_{\mathcal{H}}) \otimes I_{T_{\kappa_2}^{M_{\mathcal{N}}}T^{\mathcal{E}}_{i}T^{R}_{\mathfrak{j}} T_{\mathpzc{l}}^{N}};  \kappa_1 = \overline{0,S-1}, \kappa_2 = \overline{0,S-\kappa_1-1}, \mathfrak{j} = S, i = \overline{0,K-\kappa_1-\kappa_2}, \mathpzc{l} = 0,\\
  \end{cases}
  	 & \Bar{Y}_{\mathpzc{l}}(\kappa_1)=
   \begin{pmatrix}
      \Bar{Y}_{\mathpzc{l}}(\kappa_1,0) &  & &\\
     & \Bar{Y}_{\mathpzc{l}}(\kappa_1,1) & &\\
       & \ddots &  & \\
      & & \Bar{Y}_{\mathpzc{l}}(\kappa_1,S-\kappa_1) & 0
  \end{pmatrix},\\ &
    \Bar{Y}_0(\kappa_1,\kappa_2) = 
  \begin{cases}
 \begin{pmatrix}
      \Bar{Y}_0(1,0,0)   & \cdots& \Bar{Y}_0(1,0,S-1)&0 \\
       0&  \cdots&0& \Bar{Y}_0(1,0,S)
  \end{pmatrix},\\
      \text{diag}\{ \Bar{Y}_0(\kappa_1,\kappa_2,0),\Bar{Y}_0(\kappa_1,\kappa_2,S)\};  \text{$\forall S=K$ or $K<S~\&~ \kappa_1+\kappa_2 \leq K$,}\\
     row( \Bar{Y}_0(\kappa_1,\kappa_2,0),\Bar{Y}_0(\kappa_1,\kappa_2,S));  \text{$\forall S=K$ or $K<S ~\&~ \kappa_1+\kappa_2 = K+1$,}\\
      \Bar{Y}_0(\kappa_1,\kappa_2,0); \text{$\forall K<S~\&~ \kappa_1+\kappa_2 > K+1$,}
  \end{cases}\\
   	& \Bar{Y}_0(\kappa_1,\kappa_2,\mathfrak{j}) = \begin{cases}
\Bar{Y}_0(\kappa_1,\kappa_2,\mathfrak{j},0); ~~~ \text{$\forall \mathfrak{j} = \overline{0,S-1}$},\\
\begin{pmatrix}
      \Bar{Y}_0(\kappa_1,\kappa_2,S,0) &  & &\\
       & \ddots &  & \\
      & & \Bar{Y}_0(\kappa_1,\kappa_2,S,K-\kappa_1-\kappa_2-1)&  \Bar{Y}_0(\kappa_1,\kappa_2,S,K-\kappa_1-\kappa_2)
  \end{pmatrix}; \kappa_1\leq K,\kappa_2=0, \\
  \begin{pmatrix}
      \Bar{Y}_0(\kappa_1,\kappa_2,S,0) &  & &\\
       & \ddots &  & \\
      & & \Bar{Y}_0(\kappa_1,\kappa_2,S,K-\kappa_1-\kappa_2-1)&  0
  \end{pmatrix},  \\
  \end{cases} \\
  	 &  \Bar{Y}_\mathpzc{l}(\kappa_1,\kappa_2) = \Bar{Y}_\mathpzc{l}(\kappa_1,\kappa_2,0)= \Bar{Y}_\mathpzc{l}(\kappa_1,\kappa_2,0,0);\forall \mathpzc{l} \geq 1,\\
   & \Bar{Y}_\mathpzc{l}(\kappa_1,\kappa_2,0,0) =  
      I_{L_1L_2} \otimes L_{S-(\kappa_1+\kappa_2)}(S-\kappa_2, \Tilde{A}_{\mathcal{H}}) \otimes I_{T_{\kappa_2}^{M_{\mathcal{N}}}T_{\mathpzc{l}}^{N}}; \forall \kappa_1 = \overline{1,S}, \kappa_2 = \overline{0,S-\kappa_1},   \mathpzc{l} \geq 0,\\
   & \Bar{Y}_0(\kappa_1,\kappa_2,S,i) = I_{L_1L_2} \otimes L_{K-(\kappa_1+\kappa_2+i)}(K-\kappa_2-i, \Tilde{A}_{\mathcal{H}}) \otimes I_{T_{\kappa_2}^{M_{\mathcal{N}}}T^{\mathcal{E}}_{i}T^{R}_{S}};\forall  \kappa_1 = \overline{1,S}, \kappa_2 = \overline{0,S-\kappa_1},  i = \overline{0,K-\kappa_1-\kappa_2},\\
   & \Bar{Y}_0(\kappa_1,0,S,K-\kappa_1-\kappa_2) = C_{\mathcal{E}}\otimes I_{L_2(T_{\kappa_1}^{M_{\mathcal{H}}}\times T_{\kappa_1-1}^{M_{\mathcal{H}}})} \otimes P_i(\beta_{\mathcal{E}}) \otimes I_{T_{S}^{R}};\forall \kappa_1 = \overline{1,S}, i=K-\kappa_1.\\
  & \textbf {First column:}\\
    &\mathscr{Q}^{'}_{\mathpzc{l},0} = 
    \begin{pmatrix}
        W_\mathpzc{l}(0)&0&\cdots &0\\
        W_\mathpzc{l}(1)&0&\cdots &0\\
        \vdots&\vdots&\ddots &\vdots \\
        W_\mathpzc{l}(S)&0&\cdots &0
    \end{pmatrix};\mathpzc{l}\geq 0, ~~~~~ W_\mathpzc{l}(\kappa_1) = 
    \begin{pmatrix}
        W_\mathpzc{l}(\kappa_1,0)&0&\cdots &0\\
        W_\mathpzc{l}(\kappa_1,1)&0&\cdots &0\\
        \vdots&\vdots&\ddots& \vdots \\
        W_\mathpzc{l}(\kappa_1,S-\kappa_1)&0&\cdots &0
    \end{pmatrix}; \kappa_1 = \overline{0,S},\\
    & \text{where $\mathscr{Q}^{'}_{\mathpzc{l},0}$ and $ W_\mathpzc{l}(\kappa_1)$ are square matrices of order $S+1$ and $S-\kappa_1+1$, respectively and $\mathscr{Q}^{+} = \mathscr{Q}_{M,0}$,}\\
    & W_0(0,0) = \text{diag}\{W_0(0,0,0),W_0(0,0,1),\ldots,W_0(0,0,S)\},\\ 
    & W_0(\kappa_1,\kappa_2) = 
    \begin{cases}
      \begin{pmatrix}
        W_0(\kappa_1,\kappa_2,0) & W_0(\kappa_1,\kappa_2,1) & \cdots & W_0(\kappa_1,\kappa_2,S-1)\\
      0  & 0 &\cdots & 0 & W_0(\kappa_1,\kappa_2,S) \end{pmatrix}; \\\text{$\forall S=K$ or $K<S~\&~ \kappa_1+\kappa_2 \leq K$,}\\
      row(W_\mathpzc{l}(\kappa_1,\kappa_2,0),W_\mathpzc{l}(\kappa_1,\kappa_2,1),\ldots,W_\mathpzc{l}(\kappa_1,\kappa_2,S));\text{$\forall  K<S~\&~ \kappa_1+\kappa_2 >K$},
    \end{cases}\\
  &   W_0(\kappa_1,\kappa_2,\mathfrak{j}) = 
     \begin{cases}
      W_0(\kappa_1,\kappa_2,\mathfrak{j},0); \forall \mathfrak{j} = \overline{0,S-1},\\
       \text{diag}\{W_0(0,0,S,0),W_0(0,0,S,1),\ldots,W_0(0,0,S,K)\}\\
       \begin{pmatrix}
        W_0(\kappa_1,\kappa_2,S,0) &  &  & \\
         && \ddots &&\\
         &&& W_0(\kappa_1,\kappa_2,S,K-\kappa_1-\kappa_2-1) & 0
        \end{pmatrix}; \text{$\forall \kappa_1+\kappa_2 < K$,}\\
        row\{W_0(\kappa_1,\kappa_2,S,0),W_0(\kappa_1,\kappa_2,S,1),\ldots,W_0(\kappa_1,\kappa_2,S,K-\kappa_1-\kappa_2)\};\text{$\forall  \kappa_1+\kappa_2 \geq K$},
     \end{cases}\\
    & W_\mathpzc{l}(\kappa_1,\kappa_2,\mathfrak{j},0) = I_{L_1} \otimes D_1 \otimes e_{T^{M_{\mathcal{H}}}_{\kappa_1}T^{M_{\mathcal{N}}}_{\kappa_2}T^{M_{\mathcal{E}}}_{i}} \otimes  \Pi P_{\mathfrak{j}}(\alpha) \otimes e_{T^{N}_{\mathpzc{l}}};\forall \kappa_1 = \overline{0,S}, \kappa_2 = \overline{0,S-\kappa_1}, \mathfrak{j} = \overline{1,S}~ \& ~ \mathfrak{j}=\kappa_1+\kappa_2, \mathpzc{l} \geq 0.\\
\end{align*}}}

\subsection{Ergodicity Condition}
The modelled process  \{$\Xi(t), t \geq 0 \}$ clearly has the traits of a level-dependent quasi-birth-death process, as evidenced by its structure. But, the existing definition of $\textit{L\!D\!Q\!B\!D}$ process does not provide any results about the limiting/asymptotic behaviour of process when the countable element of Markov chain tends to infinity. Thus, the existing result of stability of constructive conditions of stability of $\textit{L\!D\!Q\!B\!D}$ process. Though, the behaviour of the proposed process satisfies the conditions imposed on the limiting behaviour of asymptotic quasi- toeplitez Markov chain ($\textrm{\it A\!Q\!T\!M\!C}$). In this order, we will compute matrices $U_0, U_1$ and $U_2$ defined as follows\\
$U_0 = \lim_{\mathpzc{l} \to \infty} T_{\mathpzc{l}}^{-1} \mathscr{Q}_{\mathpzc{l},\mathpzc{l}-1}$, $U_1 = \lim_{\mathpzc{l} \to \infty} T_{\mathpzc{l}}^{-1} \mathscr{Q}_{\mathpzc{l},\mathpzc{l}} + I$, $U_2 = \lim_{\mathpzc{l} \to \infty} T_{\mathpzc{l}}^{-1} \mathscr{Q}_{\mathpzc{l},\mathpzc{l}+1}$,\\
where $T_{\mathpzc{l}}^{-1}$ is the diagonal matrix with diagonal entries defined as the modulus of the diagonal entries of the matrix $\mathscr{Q}_{\mathpzc{l},\mathpzc{l}}, \mathpzc{l} \geq 0.$ In the considered case, the matrices $U_0, U_1$ and $U_2$ have the following form.\\
$U_0 = \lim_{\mathpzc{l} \to \infty} T_{\mathpzc{l}}^{-1} \mathscr{Q}_{\mathpzc{l},\mathpzc{l}-1}= T^{-1}\mathscr{Q}_0, \mathpzc{l}>M$,\\
$U_1 = \lim_{\mathpzc{l} \to \infty} T_{\mathpzc{l}}^{-1} \mathscr{Q}_{\mathpzc{l},\mathpzc{l}} + I= T^{-1}\mathscr{Q}_1,\mathpzc{l}>M$,\\
$U_2 = \lim_{\mathpzc{l} \to \infty} T_{\mathpzc{l}}^{-1} \mathscr{Q}_{\mathpzc{l},\mathpzc{l}+1}= T^{-1}\mathscr{Q}_2,\mathpzc{l}>M$,\\
where $T = \text{diag}\{ T(0), T(1),\ldots, T(S)\}$; $T(\kappa_1) = \text{diag}\{ T(\kappa_1,0), T(\kappa_1,1),\ldots, T(\kappa_1,S-\kappa_1)\}$,
{\small{
\begin{align*}
T(\kappa_1,\kappa_2)=
\begin{cases}
  \Lambda_0 \oplus \sum \oplus A_{\kappa_1}(S-\kappa_2, A_{\mathcal{H}}) \oplus A_{\kappa_2}(S-\kappa_1, A_{\mathcal{N}}) \oplus A_{M}(M, \gamma)
  ;& \kappa_1=\overline{0,S}, \kappa_2= \overline{0,S-\kappa_1-1},\\
     \Lambda_0+C_{\mathcal{H}} \oplus \sum \oplus A_{\kappa_1}(S-\kappa_2, A_{\mathcal{H}}) \oplus A_{\kappa_2}(S-\kappa_1, A_{\mathcal{N}}) \oplus A_{M}(M, \gamma)\\
     + I_{L_1L_2T_{\kappa_1}^{M_{\mathcal{H}}}T_{\kappa_2}^{M_{\mathcal{N}}}}\otimes L_0^{(2)}(M,\Tilde{\Gamma_2})\otimes P_M(\gamma)
  ;& \kappa_1=\overline{0,S}, \kappa_2= S-\kappa_1.\\
  \end{cases}\\
    \end{align*}}}
Here, $\Lambda_0 = -C_0$; $\sum  = -D_0 +D_1.$ Note that, in the following illustration diag, $\text{diag}^+$ and $\text{diag}^-$ represents matrix with main diagonal, upper diagonal and lower diagonal, respectively.  

\begin{theorem}
The necessary and sufficient condition for the ergodicity of the underlying process is the satisfaction of the inequality \\
\begin{align*}
\lambda <  \sum_{\kappa_1=0}^{S} \sum_{\kappa_2=0}^{S-\kappa_1} x_M^{(1)} L_0^{(1)'}(M,\Tilde{\Gamma_1})e  + \sum_{\kappa_1=0}^{S} \sum_{\kappa_2=0}^{S-\kappa_1} x_M^{(2)} L_0^{(2)'}(M,\Tilde{\Gamma_2})e
\end{align*}

where $x_M^{(1)} = x_M(\kappa_1,\kappa_2)(e_{L_2T_{\kappa_1}^{M_{\mathcal{H}}}T_{\kappa_2}^{M_{\mathcal{N}}}}\otimes I_{T^N_M})$, $x_M^{(2)} = x_M(\kappa_1,\kappa_2)(e_{L_2T_{\kappa_1}^{M_{\mathcal{H}}}T_{\kappa_2}^{M_{\mathcal{N}}}}\otimes I_{T^N_M})$ and $x$ is the unique solution of $x(\mathscr{Q}_0+\mathscr{Q}_1+\mathscr{Q}_2)=0; xe=1.$
\end{theorem}
\begin{proof}
Following \cite{klimenok2006multi}, a necessary and sufficient condition for ergodicity of the underlying process can be formulated in terms of the generator $\mathscr{Q}$ is as follows
\begin{align}
     x\mathscr{Q}_2e < x\mathscr{Q}_0e \label{eq:1}
\end{align}
   where vector $x$ is the unique solution to the system of linear equations 
   \begin{align}
        x(\mathscr{Q}_0+\mathscr{Q}_1+\mathscr{Q}_2)=0; xe=1 \label{eq:2}.
   \end{align}
  Let expression of $x$ be of the form
  $x = (\pi \otimes x_M(\kappa_1,0),\pi \otimes x_M(\kappa_1,1), \ldots, \pi \otimes x_M(\kappa_1,S-\kappa_1))$, where $x_M(\kappa_1,\kappa_2)$ vector is of size $L_2T^{M_{\mathcal{H}}}_{\kappa_1}T^{M_{\mathcal{N}}}_{\kappa_2}T^N_M.$ Substituting the expression of $x$ and block matrices in inequality \ref{eq:1}. 
  \begin{align*}
      & \sum_{\kappa_1=0}^{S}(\pi \otimes x_M(\kappa_1,S-\kappa_1)) (C_{\mathcal{N}} \otimes I_{L_2T_{\kappa_1}^{M_{\mathcal{H}}}T_{\kappa_2}^{M_{\mathcal{N}}}}\otimes P_M^{'}(\gamma) )e\\ 
      & + \sum_{\kappa_1=1}^{S}  (\pi \otimes x_M(\kappa_1,S-\kappa_1)) (C_{\mathcal{H}} \otimes I_{L_2} \otimes P_{\kappa_1}(\beta_{\mathcal{H}}) \otimes I_{T_{\kappa_2}^{M_{\mathcal{N}}}\times 
      T_{\kappa_2-1}^{M_{\mathcal{N}}}}\otimes P_M^{'}(\gamma) )e\\
      & < \sum_{\kappa_1=0}^{S} \sum_{\kappa_2=0}^{S-\kappa_1} (\pi \otimes x_M(\kappa_1,\kappa_2)) ( I_{L_1L_2T_{\kappa_1}^{M_{\mathcal{H}}}T_{\kappa_2}^{M_{\mathcal{N}}}}\otimes L_0^{(1)'}(M,\Tilde{\Gamma_1}) )e\\ 
      & + \sum_{\kappa_1=0}^{S} \sum_{\kappa_2=1}^{S-\kappa_1} (\pi \otimes x_M(\kappa_1,\kappa_2)) ( I_{L_1L_2T_{\kappa_1}^{M_{\mathcal{H}}}} \otimes P_{\kappa_2}(\beta_{\mathcal{N}})\otimes L_0^{(2)'}(M,\Tilde{\Gamma_2}) )e.
  \end{align*}
 using the relations $\pi C_{\mathcal{N}}e = \lambda_{\mathcal{N}},\pi C_{\mathcal{H}}e = \lambda_{\mathcal{H}}, \lambda_{\mathcal{N}}+\lambda_{\mathcal{N}}=\lambda, P_M^{'}(\gamma)e=e, P_{\kappa_1}(\beta_{\mathcal{H}})e=e, P_{\kappa_2}(\beta_{\mathcal{N}})e=e$, where $e$ is a column vector with one of appropriate size, the inequality will be reduced to 

\begin{align*}
\lambda <  \sum_{\kappa_1=0}^{S} \sum_{\kappa_2=0}^{S-\kappa_1} x_M^{(1)} L_0^{(1)'}(M,\Tilde{\Gamma_1})e  + \sum_{\kappa_1=0}^{S} \sum_{\kappa_2=1}^{S-\kappa_1} x_M^{(2)} L_0^{(2)'}(M,\Tilde{\Gamma_2})e.
\end{align*}
Since the stability of the system can not be defined in the catastrophic scenario, thus these conditions are derived for the normal scenario.
System of equations \ref{eq:2} has unique solution because the matrix of the system is an infinitesimal generator of the underlying process which defines joint distributions of the number of retrial calls, number of handoff calls receiving service and number of new calls receiving service. The left hand side of inequality \ref{eq:1}  is the total arrival rate of handoff calls and new calls in the system. In the right hand side, the first summand is the rate of  departure from the system  and second summand is the rate of starting the service for retrial calls when the retrial is successful. It is intuitively clear that the Markov chain describing queueing model under study is ergodic if and only if the total arrival rate is less than the maximum value of the total departure rate and successful retrial rate. When the number of retrial calls increases without bound and the retrial rate tends to infinity, the retrial queueing model approaches to the corresponding classical queueing model for which $\lambda/S\mu$ becomes a necessary and sufficient condition for the stability.
\end{proof}

\subsection{Stationary Distribution}
Let $\displaystyle{{z_s} = \{{z_s}(0), {z_s}(1),{z_s}(2),\ldots, {z_s}(M-1), {z_s}(M), \ldots \}}$ be the steady-state probability vector of generator matrix $\mathscr{Q}$ satisfying $\displaystyle{{z_s} \mathscr{Q} = 0; {z_s} e =1.}$ Here, element ${z_s}(0)$ contains \\  $\displaystyle{1 \times  \Big(\sum_{\mathfrak{j}=0}^{S} \sum_{i=0}^{K} L_1L_2 T^{M_{\mathcal{E}}}_{i}T^R_{\mathfrak{j}}\Big) \Big(\sum_{\kappa_1=0}^{S} \sum_{\kappa_2=0}^{S}\sum_{\mathfrak{j}=0}^{S}  L_1L_2 T^{M_{\mathcal{H}}}_{\kappa_1} T^{M_{\mathcal{N}}}_{\kappa_2}T^R_{S}\Big)} $ vector components and ${z_s}(\mathpzc{l})$ contains \\
 $\displaystyle{1 \times \Big(\sum_{\kappa_1=0}^{S} \sum_{\kappa_2=0}^{S}  L_1L_2 T^{M_{\mathcal{H}}}_{\kappa_1} T^{M_{\mathcal{N}}}_{\kappa_2}T^N_{\mathpzc{l}}\Big)} $ elements; $  \mathpzc{l} \geq 0,~ 0\leq \kappa_1 \leq S,~ 0\leq \kappa_2 \leq S.$
The derived structure of generator matrix lacks of the existing quasi-birth death structure and toeplitz like structure. Therefore, the existing approach for computing the stationary distribution for $\textrm{\it A\!Q\!T\!M\!C}$ can be employed because the Markov chain has a specific asymptotic behaviour. The originally proposed algorithm for $\textrm{\it A\!Q\!T\!M\!C}$ process was modified by Dudin and Dudina \cite{dudin2019retrial} later on. In their proposed approach, they tackled the challenges of larger order matrix computation and storage. Along with the substantial advantages, the proposed algorithm also has some following disadvantage.
\begin{itemize}
    \item In their algorithm, the initial value for starting the computation process has been chosen  randomly. They did not provide any reasoning on how to select the initial point $i_0.$
    \item After each unsuccessful iteration, to obtain a new searching interval, a randomly selected value $s$ has been added in the existing interval. There is no reasoning behind the chosen value $s$ also.
    \item While checking the termination criteria for matrix $F,$ they discarded some portion of the search interval at each failed iteration of Step 4.2 (case 3). They did not provided any clarification for the eliminated portion of the search interval.
    
\end{itemize}

These findings motivated for the development of a modified approach for computing the stationary distribution of the Markov chain under consideration.
\subsubsection{Old Algorithm}
\textbf{Step 1.} Set $i_0,s$ randomly. Fix $\epsilon_g$, and $\epsilon_f$ as accuracy levels of matrices $G_i$ and steady-state vector $z_s(i),$ respectively. 

\textbf{Step 2.1.}  Set $G_{\kappa}^{(1)}=O$ and $G_{\kappa}^{(2)}=I$, $\kappa=i_0-1+s,$ $i_f=-1.$ 

\textbf{Step 2.2.} Compute the matrices as defined
\begin{center}

$G_{\kappa}^{(1)}:= -(\mathscr{Q}_{{\kappa}+1,{\kappa}+1} +\mathscr{Q}_{{\kappa}+1,{\kappa}+2}  G_{{\kappa}+1}^{(1)})^{-1}\mathscr{Q}_{{\kappa}+1,{\kappa}}$,  \\
$G_{\kappa}^{(2)}:= -(\mathscr{Q}_{{\kappa}+1,{\kappa}+1} +\mathscr{Q}_{{\kappa}+1,{\kappa}+2}  G_{{\kappa}+1}^{(2)})^{-1}\mathscr{Q}_{{\kappa}+1,{\kappa}}$.
\end{center}

\textbf{Step 2.3.} Calculate $||G_{\kappa}^{(1)}-G_{\kappa}^{(2)}||.$ There can be three possible cases. 

\textit{case 1:} If $||G_{\kappa}^{(1)}-G_{\kappa}^{(2)}|| < \epsilon_g$. Go to step 3. 

\textit{case 2:} If $||G_{\kappa}^{(1)}-G_{\kappa}^{(2)}|| > \epsilon_g$ and ${\kappa} \geq i_0$. Set ${\kappa}:={\kappa}-1$ and repeat step 2.2.

\textit{case 3:} If $||G_{\kappa}^{(1)}-G_{\kappa}^{(2)}|| > \epsilon_g$ and ${\kappa} = i_0-1$. Set $s:=2s,$ ${\kappa} =i_0-1+2s,$ $i_0={\kappa},$ and repeat step 2.1.

\textbf{Step 3.} Set $G_{\kappa} = G_{\kappa}^{(1)}.$ Compute $G_i, i=(i_f+1,{\kappa}-1)$. Set $B=Q_{{\kappa},0}$ and compute $B = Q_{i,0}+ G_i B, i=(i_f+1,{\kappa}-1)$, and store.

\textbf{Step 4.1.} Set $i=i_f+1.$ If $i=0$, find solution of $z_s(0)B=0; z_s(0)e=1.$ Otherwise start $i=1$ and compute  $z_s(i)= -z_s(i-1)\mathscr{Q}_{i-1,i}(\mathscr{Q}_{i,i} +\mathscr{Q}_{i,i+1}  G_{i})^{-1}$.

\textbf{Step 4.2.} Compute $||z_s(i)||.$ There can be three possible cases

\textit{case 1:} If $||z_s(i)||< \epsilon_f$, set $i^*=i$ and go to Step 5.

\textit{case 2:} If $||z_s(i)||> \epsilon_f$, and $i<{\kappa}$, increase $i$ by one and go to Step 4.1.

\textit{case 3:} If $||z_s(i)||> \epsilon_f$, and $i={\kappa}$, set $i_0={\kappa}+s$, ${\kappa}=i_0-1+2s$ and $i_f={\kappa}$ and go to Step 2.2.

\textbf{Step 5.} Calculate vectors $z_s({\kappa}), k=\overline{0,i^*}$ as $z_s({\kappa}+1) = cz_s({\kappa})$, where $c = \frac{1}{z_s(0)+z_s(1)+z_s(2)+\ldots+z_s(i^*)}$ is a normalizing constant.

\subsubsection{New Algorithm}
\textbf{Step 1.} To fix the initial value $i_0$ (anticipated number of matrices $G_i$, which is required for the computation of the stationary distribution), convert the original model into Poisson-exponentially distributed model, i.e., arrival of calls and catastrophe are defined by Poisson process and rest of the processes follow exponential distribution. For this model, fix a pre-defined small value $\delta>0$ such that $z_s(i_0)< \delta$. Here, $z_s$ is the stationary distribution of the system and $z_s(i_0)$ is $i_0th$ component of stationary distribution. \\
\textit{Explanation:} The termination criteria or the condition $z_s(i_0)< \delta$ represents that for $i \geq i_0,$ the steady-state vector is showing invariant behaviour for a pre-defined small value. This $i_0$ will work as an initial value for the original model. In the original model also, the value of $i_0$ represents the invariant behaviour of matrix $G_i$ for $i \geq i_0$ which ultimately determines the invariant behaviour of steady-state vector. 

\textbf{Step 2.} Fix $s$ as some component of $i_0,$ i.e., $s= mi_0.$ (After many trials in the programming, it's been observed that $s=2i_0$ is appropriate choice for $s$, as the solution converges fast.)\\
\textit{Explanation:} If the Markov chain is ergodic, for any value of $i_0$, and for any fixed small arbitrary value $\epsilon$ there exist a finite number $s,$ the probability that the Markov chain will transit from $i$ to  $i+s$ without visiting state $i-1$ is less than $\epsilon.$

\textbf{Step 3.} Fix $\epsilon_g$, and $\epsilon_f$ as accuracy levels of matrices $G_i$ and steady-state vector $z_s(i),$ respectively. 

\textbf{Step 4.1.} Set $G_{\kappa}^{(1)}=O$ and $G_{\kappa}^{(2)}=I$, $\kappa=i_0-1+s,$ $i_f=-1.$ 

\textbf{Step 4.2.} Compute the matrices as defined
\begin{center}

$G_{\kappa}^{(1)}:= -(\mathscr{Q}_{{\kappa}+1,{\kappa}+1} +\mathscr{Q}_{{\kappa}+1,{\kappa}+2}  G_{{\kappa}+1}^{(1)})^{-1}\mathscr{Q}_{{\kappa}+1,{\kappa}}$,  \\
$G_{\kappa}^{(2)}:= -(\mathscr{Q}_{{\kappa}+1,{\kappa}+1} +\mathscr{Q}_{{\kappa}+1,{\kappa}+2}  G_{{\kappa}+1}^{(2)})^{-1}\mathscr{Q}_{{\kappa}+1,{\kappa}}$.
\end{center}

\textbf{Step 4.3.} Calculate $||G_{\kappa}^{(1)}-G_{\kappa}^{(2)}||.$ There can be three possible cases. 

\textit{case 1:} If $||G_{\kappa}^{(1)}-G_{\kappa}^{(2)}|| < \epsilon_g$. Go to step 5. 

\textit{case 2:} If $||G_{\kappa}^{(1)}-G_{\kappa}^{(2)}|| > \epsilon_g$ and ${\kappa} \geq i_0$. Set ${\kappa}:={\kappa}-1$ and repeat step 4.2.

\textit{case 3:} If $||G_{\kappa}^{(1)}-G_{\kappa}^{(2)}|| > \epsilon_g$ and ${\kappa} = i_0-1$. Set $s:=2s,$ ${\kappa} =i_0-1+2s,$ $i_0={\kappa},$ and repeat step 4.1.

\textbf{Step 5.} Set $G_{\kappa} = G_{\kappa}^{(1)}.$ Compute $G_i, i=(i_f+1,{\kappa}-1)$. Set $B=Q_{{\kappa},0}$ and compute $B = Q_{i,0}+ G_i B, i=(i_f+1,{\kappa}-1)$, and store.

\textbf{Step 6.1.} Set $i=i_f+1.$ If $i=0$, find solution of $z_s(0)B=0; z_s(0)e=1.$ Otherwise start $i=1$ and compute  $z_s(i)= -z_s(i-1)\mathscr{Q}_{i-1,i}(\mathscr{Q}_{i,i} +\mathscr{Q}_{i,i+1}  G_{i})^{-1}$.

\textbf{Step 6.2.} Compute $||z_s(i)||.$ There can be three possible cases

\textit{case 1:} If $||z_s(i)||< \epsilon_f$, set $i^*=i$ and go to Step 7.

\textit{case 2:} If $||z_s(i)||> \epsilon_f$, and $i<{\kappa}$, increase $i$ by one and go to Step 6.1.

\textit{case 3:} If $||z_s(i)||> \epsilon_f$, and $i={\kappa}$, set $i_0={\kappa}+1$, ${\kappa}=i_0-1+2s$ and $i_f={\kappa}$ and go to Step 4.2.

\textbf{Step 7.} Calculate vectors $z_s({\kappa}), k=\overline{0,i^*}$ as $z_s({\kappa}+1) = cz_s({\kappa})$, where $c = \frac{1}{z_s(0)+z_s(1)+z_s(2)+\ldots+z_s(i^*)}$ is a normalizing constant.

\textbf{Advantages:} The main advantage of the proposed algorithm is that the initial point will not be selected at random. Rather than starting the algorithm from the 0 or some random point, this algorithm starts the iteration with a logical initial point which saves less computation and time of matrices. The second point is that no calculation of $G$ matrix is repeated. All the matrices calculated in previous iteration are stored and used in the next iteration. In this algorithm, a sequential search is done considering a search interval. No point has been discarded unlike the previous algorithm.

\section{Performance Measures} \label{section4}
The following relevant  performance measures for  the proposed system  are calculated, after computing the  steady-state distribution $z_s$.
\begin{enumerate}

		\item The probability that there are $\mathpzc{l}$ number of retrial calls:
	\[P_{orbit}(\mathpzc{l}) = \sum_{\kappa_1=0}^{S}\sum_{\kappa_2=0}^{S-\kappa_1} {z_s} (\mathpzc{l},\kappa_1,\kappa_2,0,0)e. \]
	
	\item Expected number of retrial calls:
	\[E_{orbit} = \sum_{\mathpzc{l}=0}^{\infty} \mathpzc{l} P_{orbit}(\mathpzc{l})e. \]
	
\item The probability that $\kappa_1$ number of handoff calls are receiving service:
	\[ P_{\mathcal{E}}(\kappa_1) =   \sum_{\mathfrak{j}=1}^{S-1} {z_s} (0,0,0,\mathfrak{j},0)e + \sum_{\kappa_2=0}^{S-\kappa_1}  {z_s} (0,\kappa_1,\kappa_2,0,0)e +
\sum_{\kappa_2=0;\kappa_1+\kappa_2 \leq K}^{S-\kappa_1}  {z_s} (0,\kappa_1,\kappa_2,S,0)e \] \[+
\sum_{\kappa_2=0}^{S-\kappa_1} \sum_{i=1}^{K-\kappa_1-\kappa_2} {z_s} (0,\kappa_1,\kappa_2,S,i)e + 	\sum_{\mathpzc{l}=1}^{\infty}\sum_{\kappa_2=0}^{S-\kappa_1}  {z_s} (\mathpzc{l},\kappa_1,\kappa_2,0,0)e.\]		
		
		\item The probability that $\kappa_2$ number of new calls are receiving service:
	\[ P_{\mathcal{N}}(\kappa_2) =   \sum_{\mathfrak{j}=1}^{S-1} {z_s} (0,0,0,\mathfrak{j},0)e + \sum_{\kappa_1=0}^{S}  {z_s} (0,\kappa_1,\kappa_2,0,0)e +
\sum_{\kappa_1=0;\kappa_1+\kappa_2 \leq K}^{S}  {z_s} (0,\kappa_1,\kappa_2,S,0)e \] \[+
\sum_{\kappa_1=0}^{S} \sum_{i=1}^{K-\kappa_1-\kappa_2} {z_s} (0,\kappa_1,\kappa_2,S,i)e + 	\sum_{\mathpzc{l}=1}^{\infty}\sum_{\kappa_1=0}^{S}  {z_s} (\mathpzc{l},\kappa_1,\kappa_2,0,0)e.\]

	\item The probability that $i$ number of emergency calls are receiving service:
	\[ P_{\mathcal{E}}(i) =   \sum_{\mathfrak{j}=1}^{S} {z_s} (0,0,0,\mathfrak{j},0)e + \sum_{\kappa_1=0}^{S} \sum_{\kappa_2=0}^{S-\kappa_1} {z_s} (0,\kappa_1,\kappa_2,0,0)e +
\sum_{\kappa_1=0}^{S} \sum_{\kappa_2=0;\kappa_1+\kappa_2 \leq K}^{S-\kappa_1} {z_s} (0,\kappa_1,\kappa_2,S,0)e \] \[+
\sum_{\kappa_1=0}^{S} \sum_{\kappa_2=0;\kappa_1+\kappa_2 + i \leq K}^{S-\kappa_1} {z_s} (0,\kappa_1,\kappa_2,S,i)e + 	\sum_{\mathpzc{l}=1}^{\infty}\sum_{\kappa_1=0}^{S}\sum_{\kappa_2=0}^{S-\kappa_1}  {z_s} (\mathpzc{l},\kappa_1,\kappa_2,0,0)e.\]

	\item The probability that the system is under repair:
	\[ P_R =   \sum_{\mathfrak{j}=1}^{S-1} {z_s} (0,0,0,\mathfrak{j},0)e + \sum_{\kappa_1=0}^{K} \sum_{\kappa_2=0}^{K-\kappa_1} \sum_{i=0}^{K-\kappa_1-\kappa_2} {z_s} (0,\kappa_1,\kappa_2,S,i)e. \]
	
		\item The dropping probability of a handoff call:
\begin{itemize}
    \item[-] in normal scenario:
    	\[ P_d^n =  \frac{1}{\lambda_{\mathcal{H}}} \Big( \sum_{\mathpzc{l}=0}^{\infty} {z_s} (\mathpzc{l},S,0,0,0) (C_{\mathcal{H}}\otimes I_{
	{\scriptstyle{L_2T^{M_{\mathcal{H}}}_{\kappa_1}T^{N}_{\mathpzc{l}}}}}) e\Big) \]
	\item[-] in catastrophic scenario:
	\[ P_d^c =  \frac{1}{\lambda_{\mathcal{N}}} \Big( \sum_{\kappa_1=0}^{K} \sum_{\kappa_2=0}^{K-\kappa_1}{z_s} (0,\kappa_1,\kappa_2,S,K-\kappa_1-\kappa_2) (C_{\mathcal{H}}\otimes I_{
	{\scriptstyle{L_2T^{M_{\mathcal{H}}}_{\kappa_1}T^{M_{\mathcal{N}}}_{\kappa_2}T^{M_{\mathcal{E}}}_{K-\kappa_1-\kappa_2}T^{R}_{S}}}}) e\Big). \]
	
\end{itemize}

	\item The blocking probability of a new call in catastrophic scenario:
	\[ P_b^c =  \frac{1}{\lambda_{\mathcal{N}}} \Big( \sum_{\kappa_1=0}^{K} \sum_{\kappa_2=0}^{K-\kappa_1}{z_s} (0,\kappa_1,\kappa_2,S,K-\kappa_1-\kappa_2) (C_{\mathcal{N}}\otimes I_{
	{\scriptstyle{L_2T^{M_{\mathcal{H}}}_{\kappa_1}T^{M_{\mathcal{N}}}_{\kappa_2}T^{M_{\mathcal{E}}}_{K-\kappa_1-\kappa_2}T^{R}_{S}}}}) e\Big). \]

	\item The blocking probability of an emergency call:
	\[ P_e =  \frac{1}{\lambda_{\mathcal{E}}} \Big( {z_s} (0,0,0,S,K) (C_{\mathcal{E}}\otimes I_{
	{\scriptstyle{L_2T^{M_{\mathcal{E}}}_{K}T^{R}_{S}}}}) e\Big). \]

    \item Rate of losses due to the occurrence of catastrophe:
	\[\alpha_{f} = \alpha \sum_{\mathpzc{l}=0}^{\infty} \sum_{\kappa_1=0}^{S} \sum_{\kappa_2=0; \kappa_1=\kappa_1 \neq 0}^{S-\kappa_1}  {z_s} (\mathpzc{l},\kappa_1,\kappa_2,0,0)  (D_1\otimes I_{
	{\scriptstyle{L_1T^{M_{\mathcal{H}}}_{\kappa_1}T^{M_{\mathcal{N}}}_{\kappa_2}T^{N}_{\mathpzc{l}}}}})  e. \]     
	
	\item The probability that an arriving handoff call preempts the service of an ongoing new call in ordinary scenario:
	\[P_{preempt}^{new} = \frac{1}{\lambda_{\mathcal{H}}} \sum_{\mathpzc{l}=0}^{\infty}\sum_{\kappa_1=0}^{S-1}  {z_s} (\mathpzc{l},\kappa_1,S-\kappa_1,0,0)  (C_{\mathcal{H}}\otimes I_{
	{\scriptstyle{L_2T^{M_{\mathcal{H}}}_{\kappa_1}T^{M_{\mathcal{N}}}_{S-\kappa_1}T^{N}_{\mathpzc{l}}}}})  e. \]  
	
	\item The probability that an arriving emergency call preempts the service of an ongoing handoff/new call in catastrophic scenario:
	\[P_{preempt}^{emr} = \frac{1}{\lambda_{\mathcal{E}}} \Big(\sum_{\kappa_2=1}^{K}  {z_s} (0,0,\kappa_2,S,K-\kappa_2) + \sum_{\kappa_1=1}^{K} \sum_{\kappa_2=0}^{K-\kappa_1}  {z_s} (0,\kappa_1,\kappa_2,S,K-\kappa_1-\kappa_2) \Big)  \Big(C_{\mathcal{E}}\otimes I_{
	{\scriptstyle{L_2T^{M_{\mathcal{H}}}_{\kappa_1}T^{M_{\mathcal{N}}}_{\kappa_2}T^{M_{\mathcal{E}}}_{K-\kappa_1-\kappa_2}T^{R}_{S}}}} \Big) e. \]   
	
	\item The intensity by which a retrial call is successfully connected to an available channel:
	\[ \theta_r^{succ} =  \sum_{\mathpzc{l}=1}^{\infty}\sum_{\kappa_1=0}^{S}\sum_{\kappa_2=0}^{S-\kappa_1} \theta {z_s} (\mathpzc{l},\kappa_1,\kappa_2,0,0) ( I_{
	{\scriptstyle{L_1L_2T^{M_{\mathcal{H}}}_{\kappa_1}T^{M_{\mathcal{N}}}_{\kappa_2}T^{N}_{\mathpzc{l}-1}}}} \otimes( \Gamma^{0}(2)\otimes \beta_{\mathcal{N}}))e.\]

\end{enumerate}
\section{Numerical Illustration} \label{section5}
In this section, the qualitative behaviour of the proposed model is explored through a few experiments. All the numerical experiments have been conducted considering $\delta=10^{-12}, \epsilon_g=10^{-10}, $ and $\epsilon_f=10^{-10}$. 
 For the numerical computation,  the matrices for the $\textrm{\it M\!M\!A\!P}$ are referred from \cite{dudin2016analysis} as follows in normal scenario,
\begin{align}
	\nonumber
	C_{0}= \begin{pmatrix}
		-0.8109843 & 0\\
		0 & -0.02632213
	\end{pmatrix},C_{\mathcal{H}}= \begin{pmatrix}
		0.201398 & 0.0013479\\
		0.003665 & 0.0029153
	\end{pmatrix}, C_{\mathcal{N}}= \begin{pmatrix}
	0.604194 & 0.004043935\\
		0.0109956 & 0.00874592
	\end{pmatrix}.
\end{align}
and in catastrophic scenario as follows,
\begin{align}
	\nonumber
	C_{0}= \begin{pmatrix}
		-0.810 & 0\\
		0 & -0.026
	\end{pmatrix},C_{\mathcal{H}}= \begin{pmatrix}
		0.20 & 0.0013\\
		0.003 & 0.002
	\end{pmatrix}, C_{\mathcal{N}}= \begin{pmatrix}
	0.30 & 0.0020\\
		0.005 & 0.004
	\end{pmatrix},C_{\mathcal{E}}= \begin{pmatrix}
		0.30 & 0.0020\\
		0.005 & 0.004
	\end{pmatrix}.
\end{align}
The  correlation coefficients for both types of calls are $C_{r}^{(1)}=C_{r}^{(2)}=0.2$ and the  variation coefficients  for both types of calls are $C_{r}^{(1)}=C_{r}^{(2)}=12.34.$ 
The arrival rates  $\lambda_{\mathcal{H}}= 0.15, \lambda_{\mathcal{N}} =0.45$.
Let $\textrm{\it P\!H}$ distributions parameters for the service rates of a handoff and a new call  be
\begin{align}
	\nonumber
	\beta_{\mathcal{H}}= \begin{pmatrix}
	0.05, &0.95
	\end{pmatrix},~~ A_{\mathcal{H}}= \begin{pmatrix}
		-0.031 & 0\\
		0 & -2.4
	\end{pmatrix},~~\textrm{and}~~~\beta_{\mathcal{N}}= \begin{pmatrix}
		0.1,& 0.9
	\end{pmatrix}, ~~ A_{\mathcal{N}}= \begin{pmatrix}
		-0.033 & 0\\
		0 & -2.52
	\end{pmatrix},
\end{align}

\begin{align}
	\nonumber
\textrm{and}~~~	\beta_{\mathcal{E}}= \begin{pmatrix}
	0, &1
	\end{pmatrix},~~ A_{\mathcal{E}}= \begin{pmatrix}
		-1 & 0\\
		0 & -1
	\end{pmatrix}.
\end{align}
The fundamental service rates are  $\mu_{\mathcal{H}}=0.5,\mu_{\mathcal{N}}=0.3$. The retrial rate of a retrial call, following  $\textrm{\it PH}$ distribution, is given by the parameters \cite{artalejo2007modelling}
\begin{align*}
	\nonumber
	\gamma= \begin{pmatrix}
		1, & 0
	\end{pmatrix},~~~ \Gamma = \begin{pmatrix}
		-2 & 2\\
		0 & -2
	\end{pmatrix},~~~~\theta=1.
\end{align*}

To demonstrate the feasibility of the developed model,
some interesting observations of the proposed system are described through the following numerical experiments. These experiments will present the behaviour of  performance measures with respect to  arrival, service and retrial rates.\\

\textbf{Experiment 1:} The objective here is to analyze the impact of arrival rate ($\lambda_{\mathcal{H}}$) and service rate ($\mu_{\mathcal{H}}$) of handoff call  over the loss probabilities, i.e., dropping probability of handoff call in normal scenario ($P_d^n$) and preemption probability for new call in normal scenario ($P_{preempt}^{new}$).

Figures \ref{fig:1a} and \ref{fig:1b} represent $P_d^n$ and $P_{preempt}^{new}$ as functions of total number of channels $S$ and  $\lambda_{\mathcal{H}}$.  It can be observed from Figure \ref{fig:1a} that $P_d^n$  increases with respect to $\lambda_{\mathcal{H}}$ for a fixed value of $S$. Moreover, under the same value of $\lambda_{\mathcal{H}}$, $P_d^n$  decreases with respect to $S$. When  handoff calls arrive frequently in the system, all the channels are most likely to be occupied by handoff calls and consequently extra handoff calls will be dropped.
Therefore,  the value of $P_d^n$  increases with increasing $\lambda_{\mathcal{H}}$. In this scenario, if the total channels are increased in the system, more handoff call will be able to obtain the service, as a result, $P_d^n$  decreases.
Figure \ref{fig:1b} exhibits the impact of  $\lambda_{\mathcal{H}}$ and $S$ over $P_{preempt}^{new}$.
It is seen from the graph that the values of $P_{preempt}^{new}$  for different $S$, first increase, and then decrease. The cause for this behavior of $P_{preempt}^{new}$  lies in the following explanation.
When $\lambda_{\mathcal{H}}$  is relatively small, an arriving handoff call often finds at least one channel available, and consequently the ongoing service of a new call is not preempted by the arriving handoff call. As $\lambda_{\mathcal{H}}$ increases, the number of handoff calls also increase in the system. If an  arriving handoff call finds all the channels occupied and at least one of them is serving a new call, the service of that new call will be preempted by the arriving handoff call. Hence, $P_{preempt}^{new}$  increases and reaches maximum at some value of $\lambda_{\mathcal{H}}$. Further, the decreasing behaviour of $P_{preempt}^{new}$   is explained by the fact that, with the increment in $\lambda_{\mathcal{H}}$,  all the channels are occupied with handoff calls. Thus, the number of new calls in the service decreases and the probability that an arriving handoff call preempts the service of a new call decreases. 

Figures \ref{fig:2a} and \ref{fig:2b} show the decreasing behaviour of  $P_d^n$ and $P_{preempt}^{new}$ as  $\mu_{\mathcal{H}}$ increases in the system. Further, it can also be observed from  Figures \ref{fig:2a} and \ref{fig:2b} that, for the fixed value of $\mu_{\mathcal{H}}$,  $P_d^n$ and $P_{preempt}^{new}$ decrease as  $S$ increases. An intuitive explanation for this finding can easily be given as follows. As $\mu_{\mathcal{H}}$ increases, the calls are served with increasing rate, consequently the loss probabilities for both types of calls  decrease. If $S$ is increased in the system, the chances  for both types of calls to obtain the service also increase and consequently $P_d^n$and $P_{preempt}^{new}$ decrease.
\\

\textbf{Experiment 2:} The main motivation of this experiment is to show how the consideration of exponential retrial times might lead to under/over estimation of  performance measures of the system.

Figures \ref{fig:6a} and \ref{fig:6b} exhibits the behaviour of the intensity by which a retrial call is successfully connected to an available channels $\theta_r^{succ}$ with respect to  retrial rate $\theta$ and expected number of retrial calls $E_{orbit}$ with respect to new call arrival rate $\lambda_{\mathcal{N}}$. It can be seen from the graphs that $\theta_r^{succ}$ increases as $\theta$ increases in the system because the probability that more retrial calls will get connection increases. A clear difference can be observed in the values of  $\theta_r^{succ}$ for exponential retrial case and $P\!H$ distributed retrial case. The value of $\theta_r^{succ}$ is high for $P\!H$ distributed retrial case in comparison to exponentially distributed retrial case. This finding is useful in determining how taking exponentiality into account in the case of the retrial phenomenon can lead to under/over estimation of system performance. The retrial phenomena that occurs in communication networks is certainly not an exponential process, as is well acknowledged. Since the retrial of  call operation completes in many phases. Therefore, $P\!H$ distribution is an apt representation. The obtained results justifies the above mentioned statement. The same increasing behaviour can be observed from the Figure \ref{fig:6b} for $E_{orbit}$ with respect to $\lambda_{\mathcal{N}}$. The value of $E_{orbit}$ is more in case of $P\!H$ distributed retrial case than the exponentially distributed.\\

\textbf{Experiment 3:} The main purpose of this experiment is to observe the behaviour of loss probabilities in catastrophic scenario with respect to the arrival rate of emergency call ($\lambda_{\mathcal{E}}$) and service rate of emergency call ($\mu_{\mathcal{E}}$). 

Figures \ref{fig:3a} and \ref{fig:3b} exhibit the behaviour of blocking probability for an emergency call $P_{\mathcal{E}}$ with respect to  $\lambda_{\mathcal{E}}$ and $\mu_{\mathcal{E}}$.  $P_{\mathcal{E}}$ increases with respect to $\lambda_{\mathcal{E}}$ and decreases with respect to $\mu_{\mathcal{E}}$. An intuitive explanation for this finding can easily be given as follows. As $\lambda_{\mathcal{E}}$ increases, more emergency calls arrive to the system which leads to an increment in the blocking probability of emergency calls. The decreasing behaviour can be observed as $\mu_{\mathcal{E}}$ increases in the system because the emergency calls are served with increasing rate, consequently $P_{\mathcal{E}}$ decreases. If $K$ is increased in the system, the chances  for emergency calls to obtain the service also increase and consequently $P_{\mathcal{E}}$ decreases.

Figures \ref{fig:4a} and \ref{fig:4b} show the behaviour of blocking probability of new call $P_b^c$ with respect to $\lambda_{\mathcal{E}}$ and $\mu_{\mathcal{E}}$. $P_b^c$ increases as  $\lambda_{\mathcal{E}}$ increases in the system and decreases as $\mu_{\mathcal{E}}$ increases. As the number of emergency calls increases in the system, less channels will be available for handoff and new calls. As a consequence, an arriving handoff/new call will be dropped from the system and $P_b^c$ increases. The decreasing behaviour of $P_b^c$ can be observed as $\mu_{\mathcal{E}}$ increases in the system because the emergency calls will be served with an increasing rate, and more channels will be available for other calls, consequently $P_b^c$ decreases. If $K$  increases in the system, the chances  for handoff/new calls to obtain the service also increase and consequently $P_b^c$ decreases.

Figures \ref{fig:5a} and \ref{fig:5b} represents the behaviour of preemption probability for emergency call $P_{preempt}^{emr}$ over $\lambda_{\mathcal{E}}$ and $\mu_{\mathcal{E}}$. Initially, for small values of  $\lambda_{\mathcal{E}}$, $P_{preempt}^{emr}$ increases which shows that  due to the arrival of emergency calls, more handoff/new calls will be preempted from the system. After a certain value of $\lambda_{\mathcal{E}}$, $P_{preempt}^{emr}$ starts decreasing. Since, after a certain time period, most of the channels will be occupied with the emergency calls only, consequently, there will be no preemption of handoff/new calls in the system. Therefore, $P_{preempt}^{emr}$ decreases as $\lambda_{\mathcal{E}}$ increases. It can be observed from Figure \ref{fig:5b} that $P_{preempt}^{emr}$ is a decreasing function of $\mu_{\mathcal{E}}$ which is obvious. These observations are the main motivation for the formulation of the  optimization problem illustrated in Section \ref{section6}.

\begin{figure}[htp]
	\centering
	\subfigure[$P_d^n$ versus  $\lambda_{\mathcal{H}}$]
	{\includegraphics[trim= 3cm 0.1cm 3.0cm 0.4cm, height = 5.55cm,width = 0.5\textwidth]{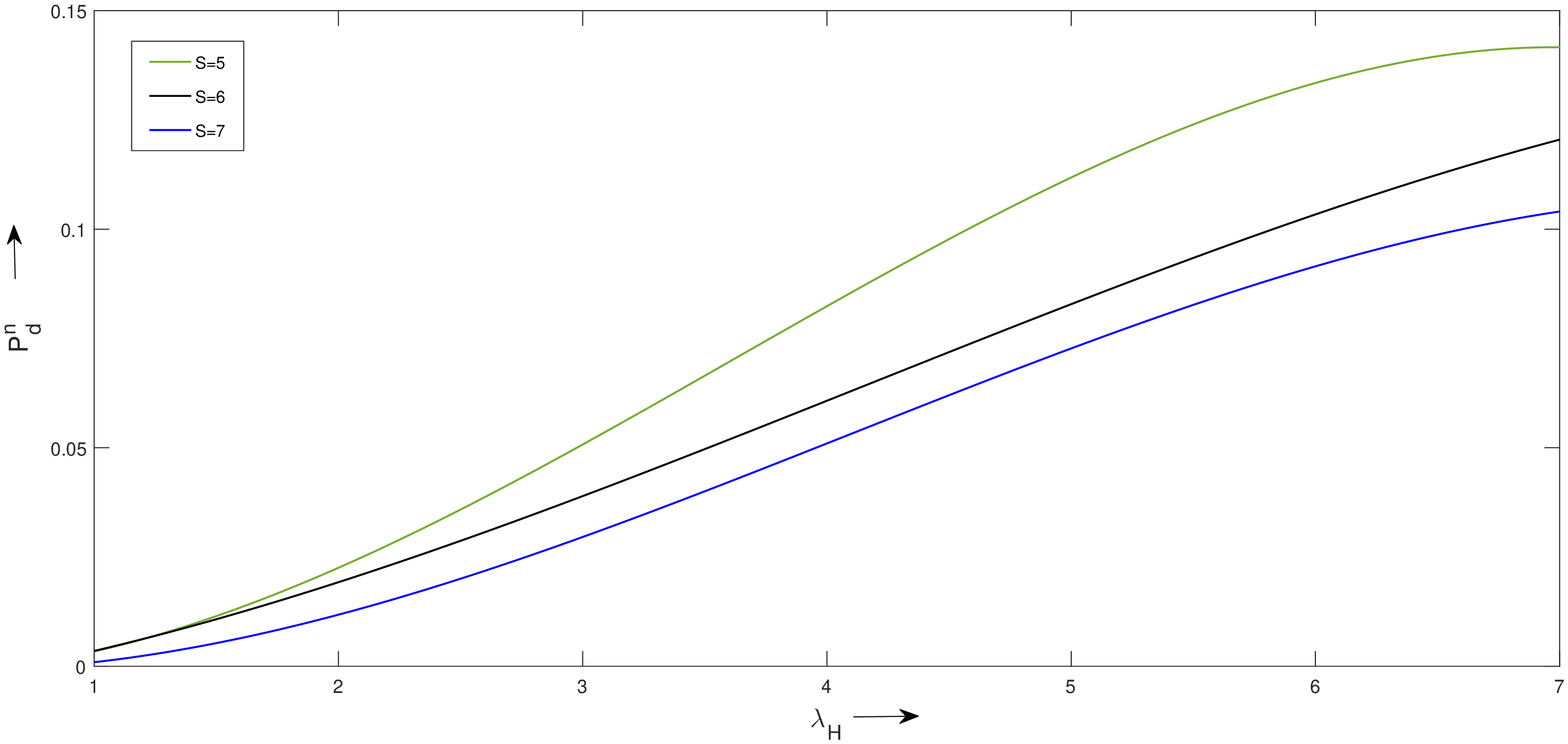}
		\label{fig:1a}}%
	\subfigure[$P_{preempt}^{new}$ versus  $\lambda_{\mathcal{H}}$]
	{\includegraphics[trim= 2.5cm 0.05cm 1.5cm 0.3cm, clip=true, height = 5.5cm,width = 0.5\textwidth]{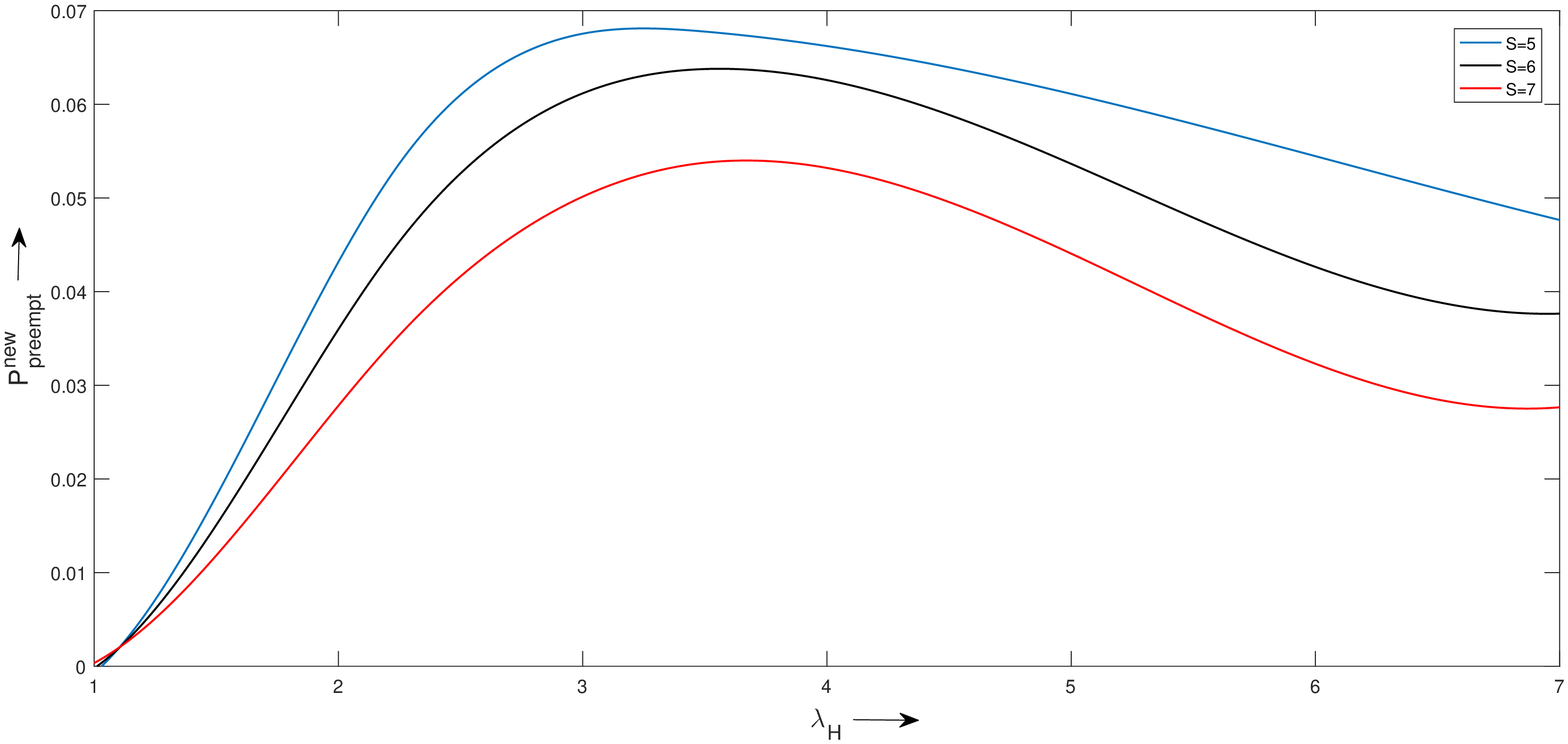}
		\label{fig:1b}}%
	\caption{Dependence of the dropping probability $P_d^n$ and preemption probability $P_{preempt}^{new}$ over arrival rate of a handoff call $\lambda_{\mathcal{H}}$. }
		\subfigure[$P_d^n$ versus  $\mu_{\mathcal{H}}$]
	{\includegraphics[trim= 3cm 0.1cm 3.0cm 0.4cm, height = 5.55cm,width = 0.5\textwidth]{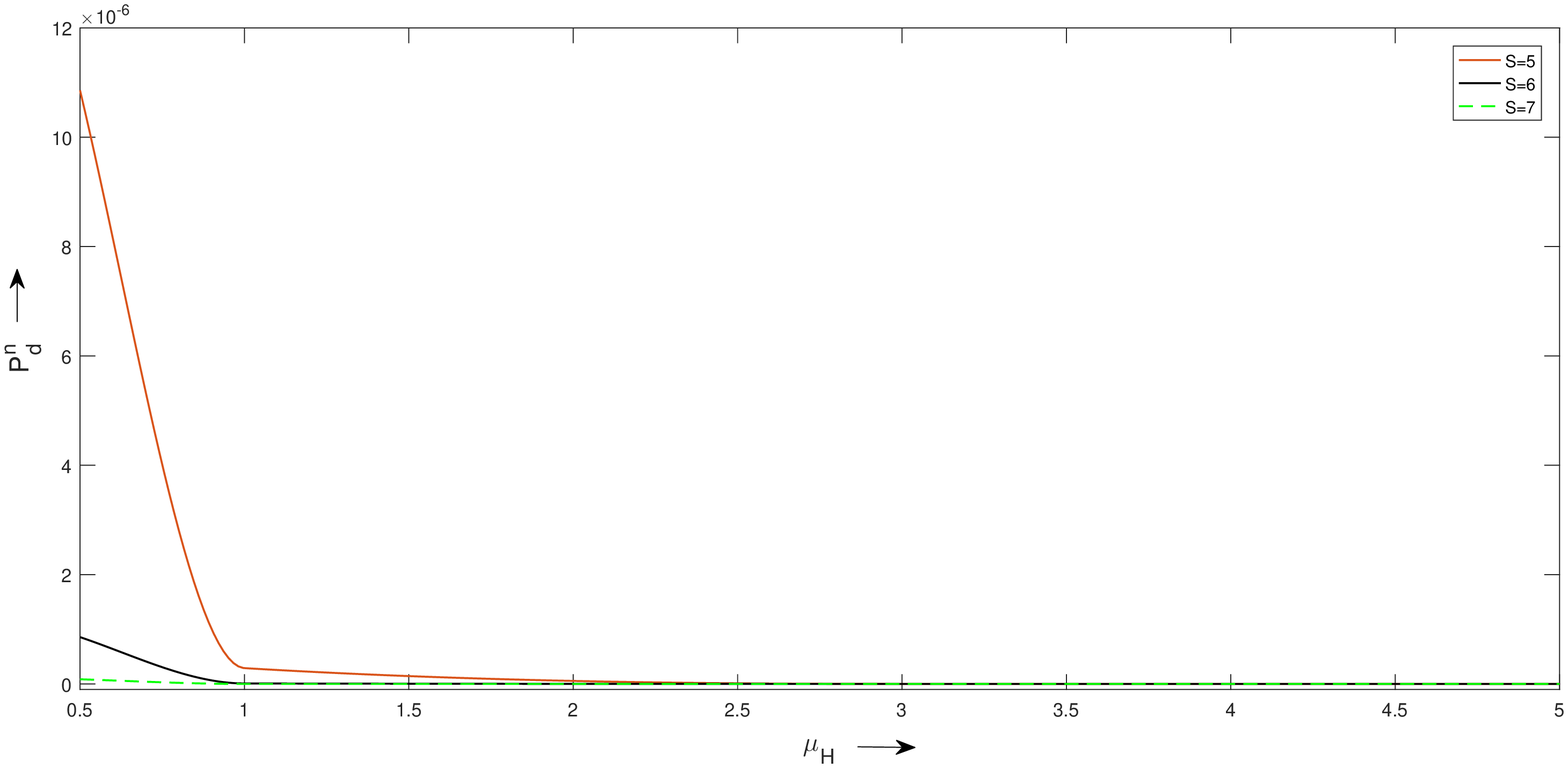}
		\label{fig:2a}}%
	\subfigure[$P_{preempt}^{new}$ versus  $\mu_{\mathcal{H}}$]
	{\includegraphics[trim= 2.5cm 0.05cm 2cm 0.3cm, clip=true, height = 5.5cm,width = 0.5\textwidth]{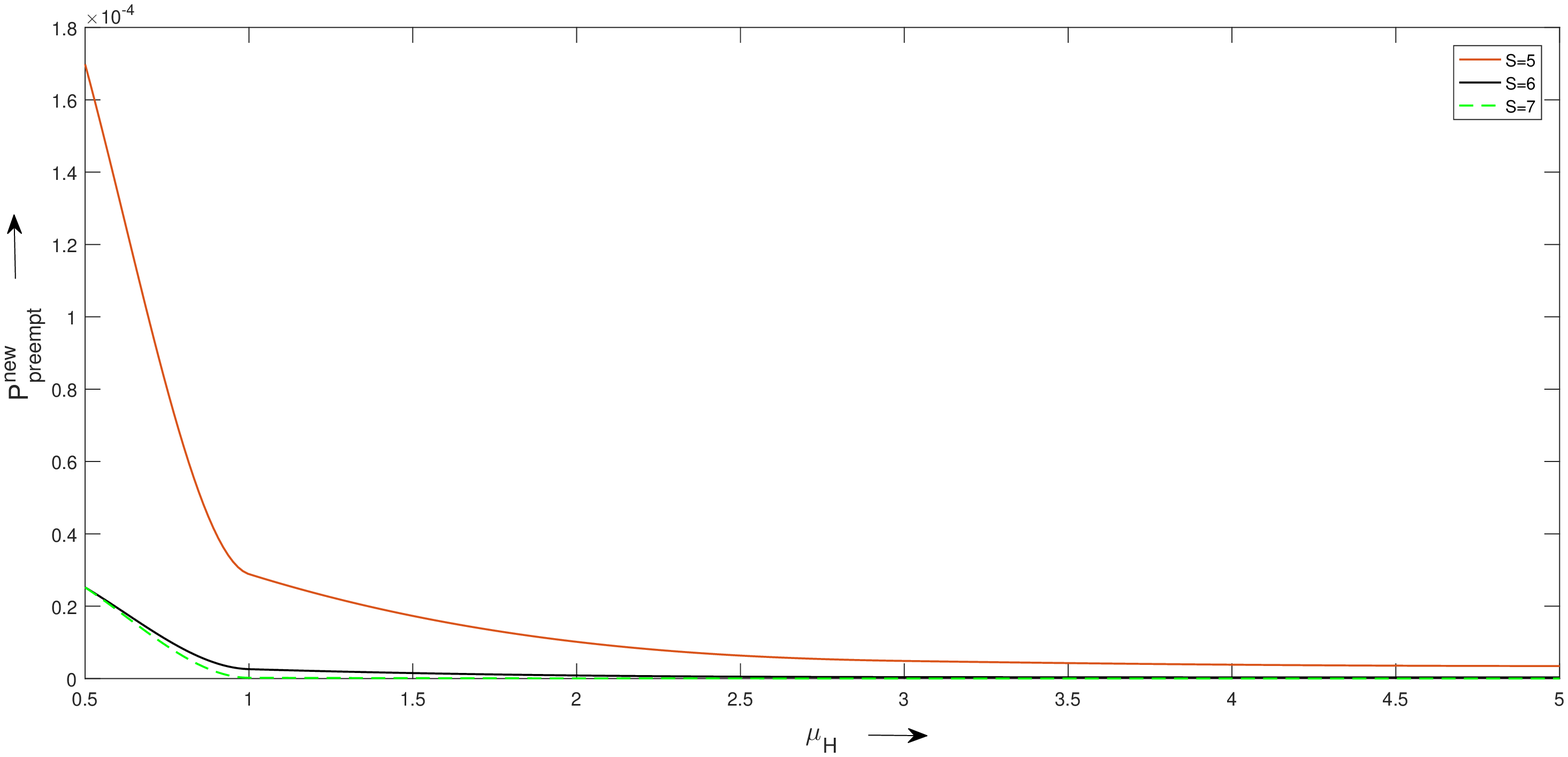}
		\label{fig:2b}}%
	\caption{Dependence of the dropping probability $P_d^n$ and preemption probability $P_{preempt}^{new}$ over service rate of a handoff call $\mu_{\mathcal{H}}$. }
		\subfigure[$\theta_r^{succ}$ versus  $\theta$]
	{\includegraphics[trim= 3cm 0.1cm 3.0cm 0.4cm, height = 5.55cm,width = 0.5\textwidth]{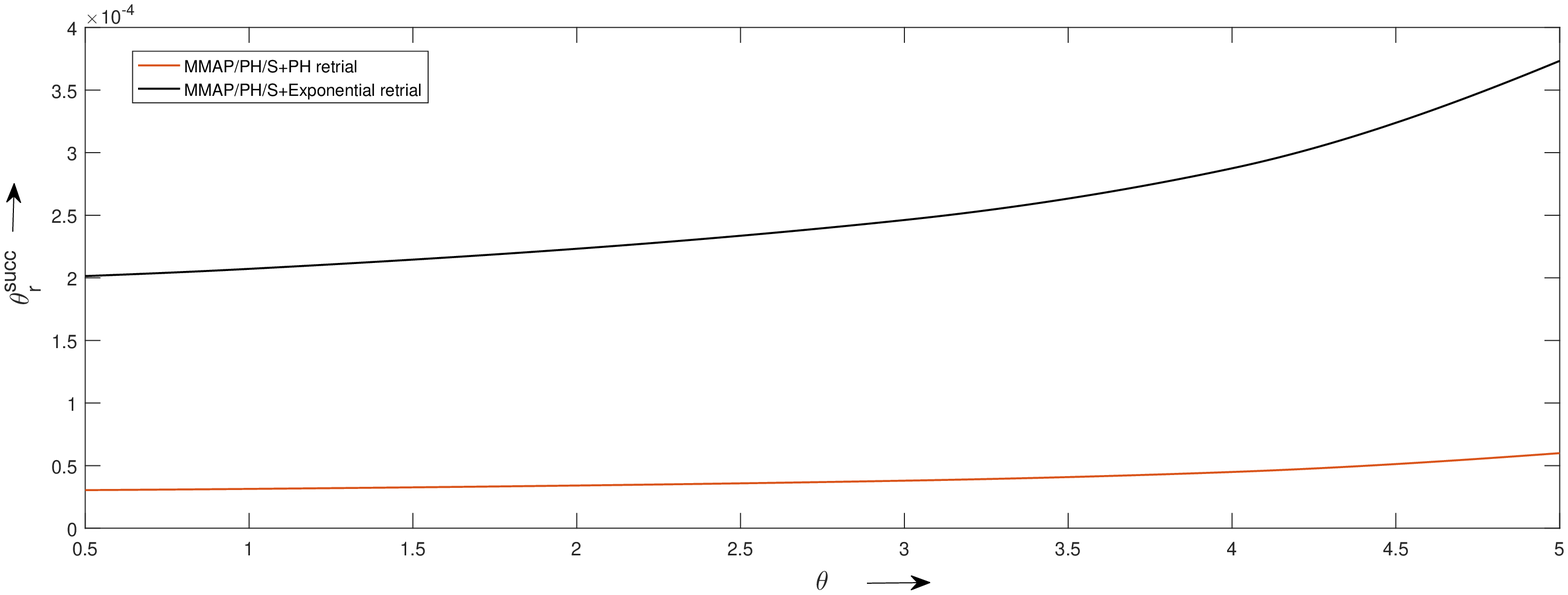}
		\label{fig:6a}}%
	\subfigure[$E_{orbit}$ versus  $\lambda_{\mathcal{N}}$]
	{\includegraphics[trim= 2.5cm 0.05cm 1.5cm 0.3cm, clip=true, height = 5.9cm,width = 0.5\textwidth]{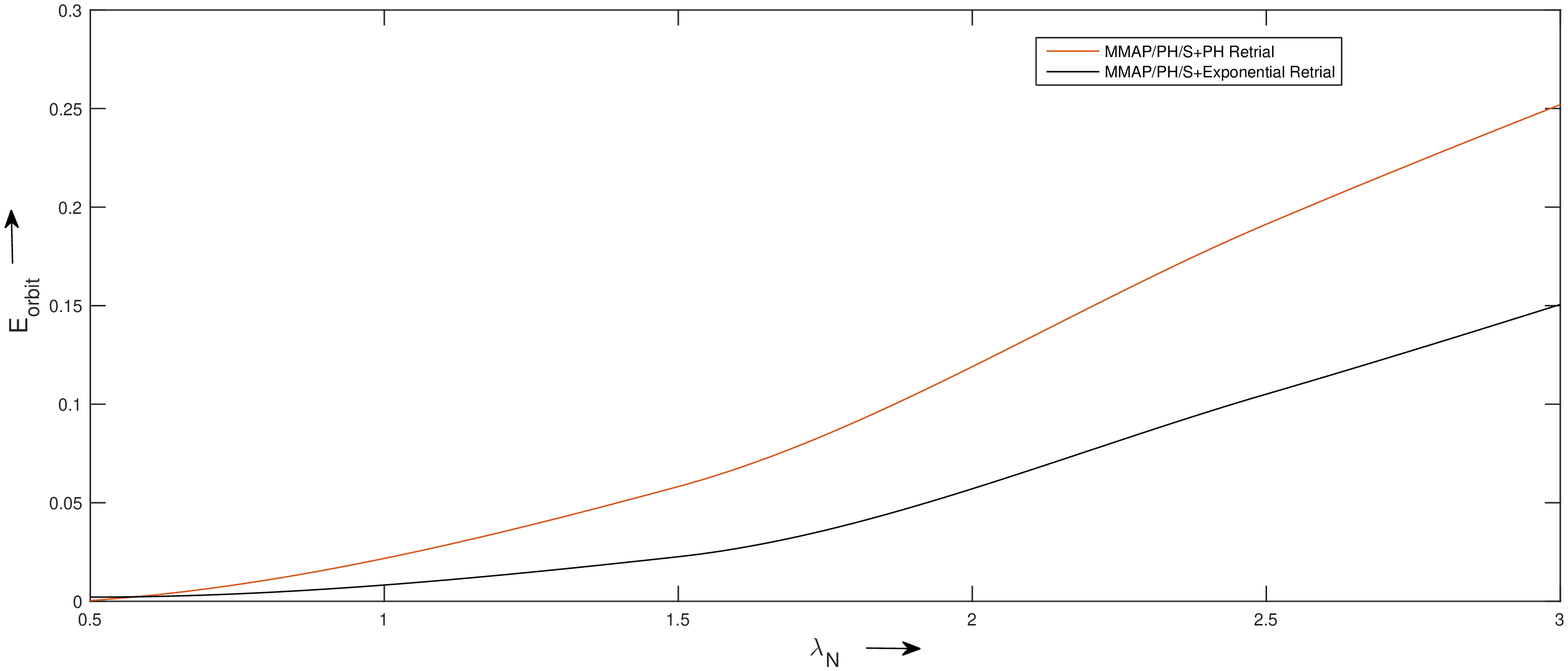}
		\label{fig:6b}}%
	\caption{Dependence of the intensity by which a retrial call is successfully connected to an available channel $\theta_r^{succ}$ over $\theta$ and  expected number of retrial calls $E_{orbit}$ over arrival rate of a new call $\lambda_{\mathcal{N}}$. }
	\label{fig:Pb_Pd1}
\end{figure}

\begin{figure}[htp]
	\centering
	\subfigure[$P_e$ versus  $\lambda_{\mathcal{E}}$]
	{\includegraphics[trim= 3cm 0.1cm 3.0cm 0.4cm, height = 5.55cm,width = 0.5\textwidth]{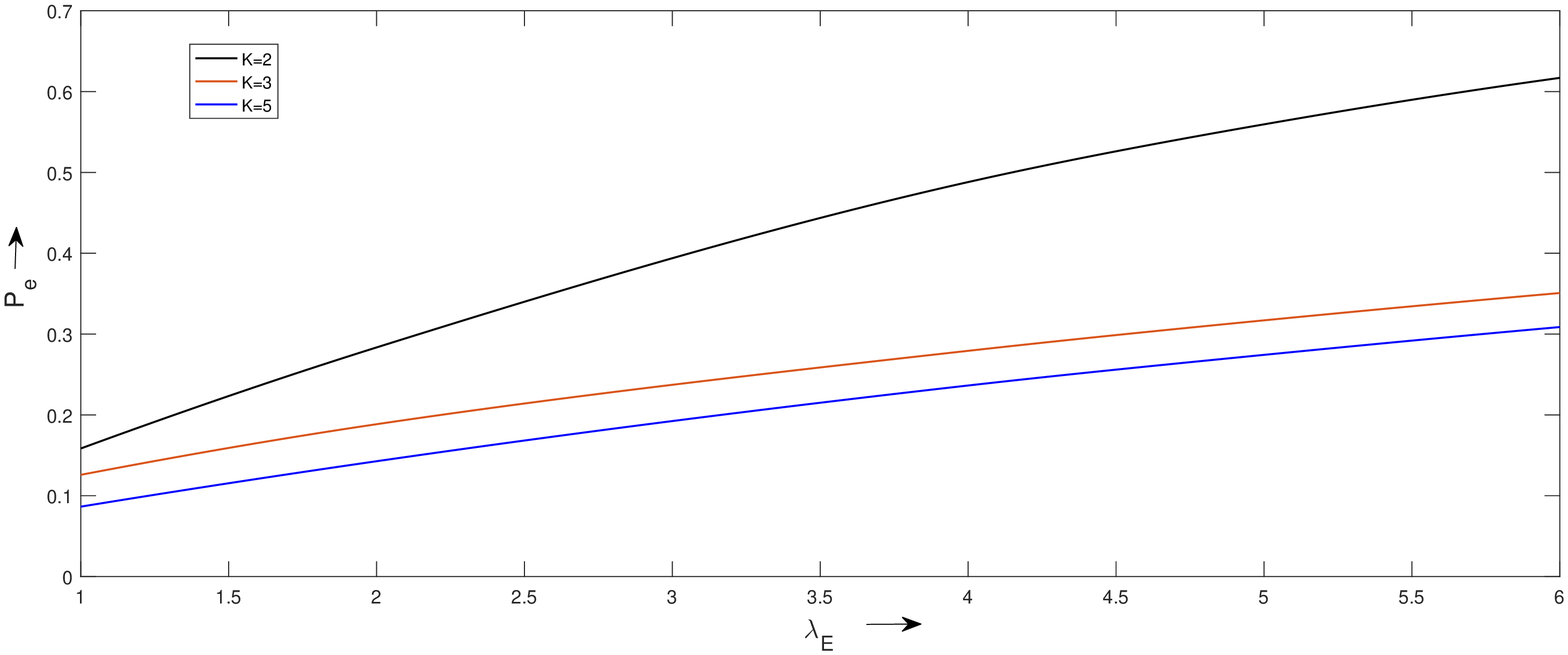}
		\label{fig:3a}}%
	\subfigure[$P_e$ versus  $\mu_{\mathcal{E}}$]
	{\includegraphics[trim= 2.5cm 0.05cm 1.5cm 0.3cm, clip=true, height = 5.5cm,width = 0.5\textwidth]{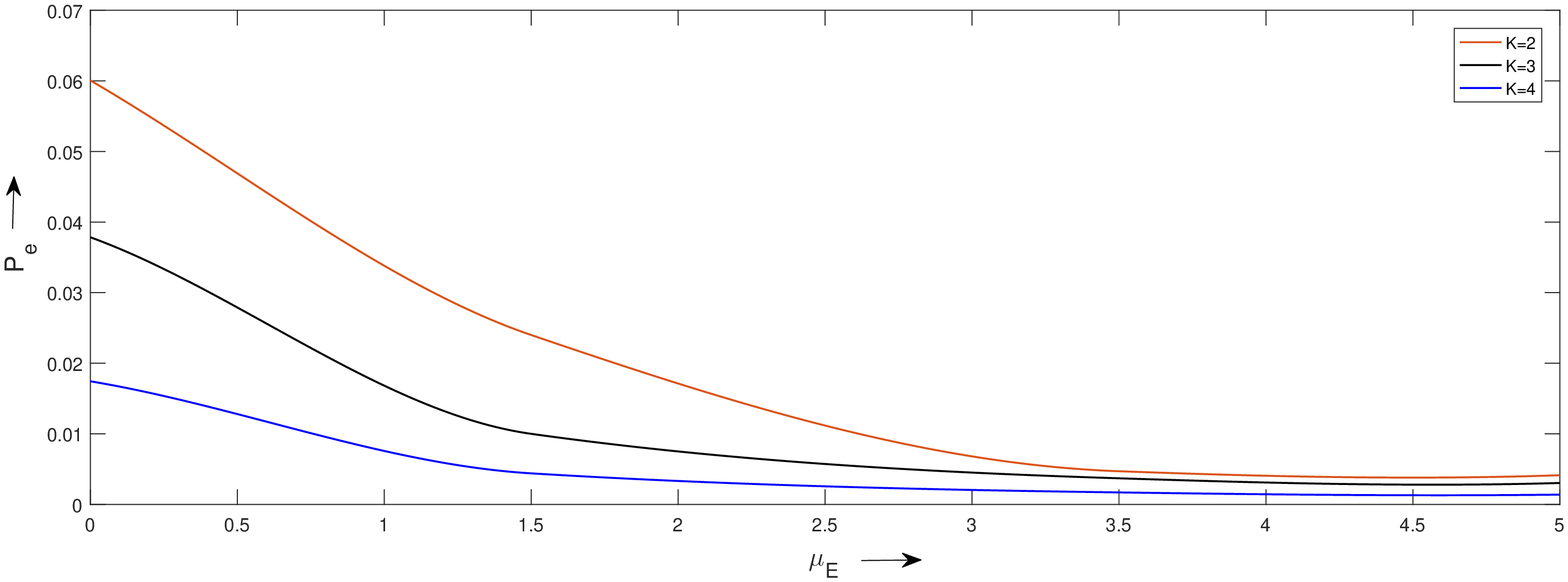}
		\label{fig:3b}}%
	\caption{Dependence of the  blocking probability for emergency call $P_e$ over arrival rate of an emergency call $\lambda_{\mathcal{E}}$ and service rate of an emergency call $\mu_{\mathcal{E}}$. }
		\subfigure[$P_b^c$ versus  $\lambda_{\mathcal{E}}$]
	{\includegraphics[trim= 3cm 0.1cm 3.0cm 0.4cm, height = 5.55cm,width = 0.5\textwidth]{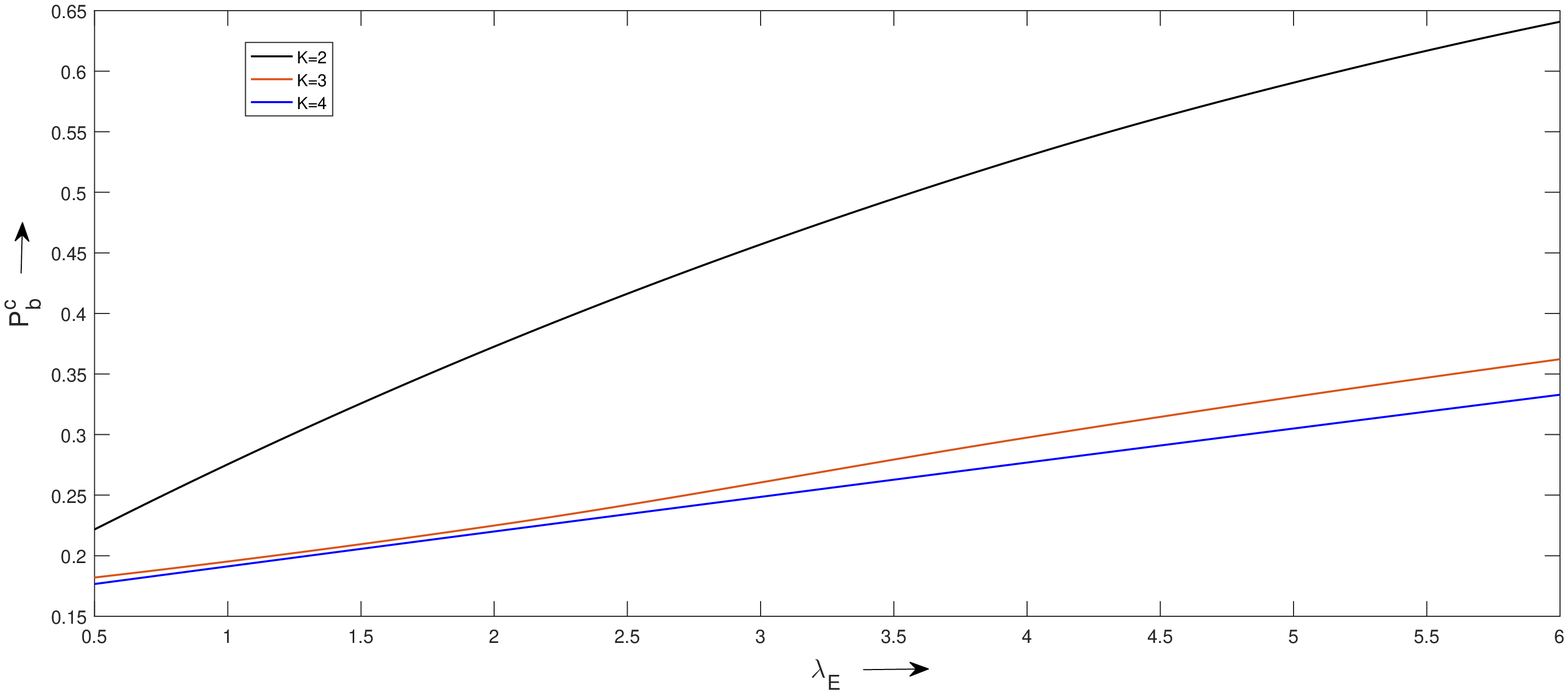}
		\label{fig:4a}}%
	\subfigure[$P_{b}^{c}$ versus  $\mu_{\mathcal{E}}$]
	{\includegraphics[trim= 2.5cm 0.05cm 2cm 0.3cm, clip=true, height = 5.5cm,width = 0.5\textwidth]{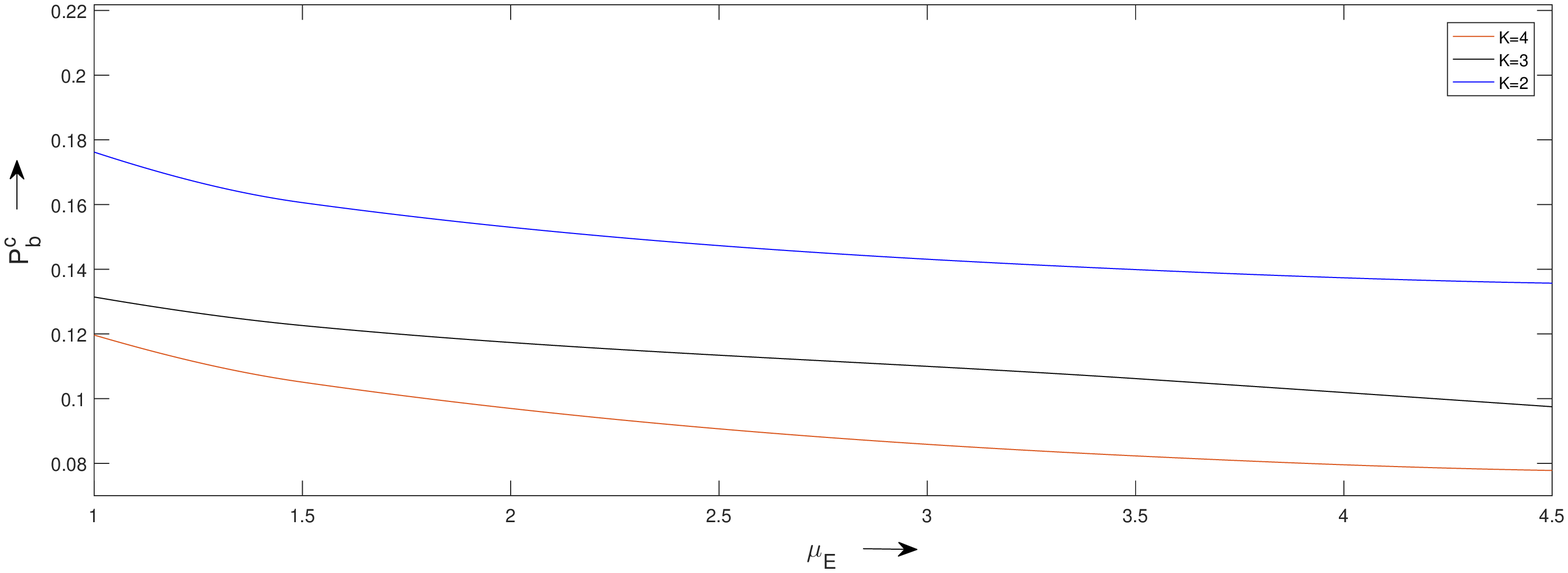}
		\label{fig:4b}}%
	\caption{Dependence of the blocking probability of new call $P_b^c$  over arrival rate of an emergency call $\lambda_{\mathcal{E}}$ and service rate of an emergency call $\mu_{\mathcal{E}}$. }
		\subfigure[$P_{preempt}^{emr}$ versus  $\lambda_{\mathcal{E}}$]
	{\includegraphics[trim= 3cm 0.1cm 3.0cm 0.4cm, height = 5.55cm,width = 0.5\textwidth]{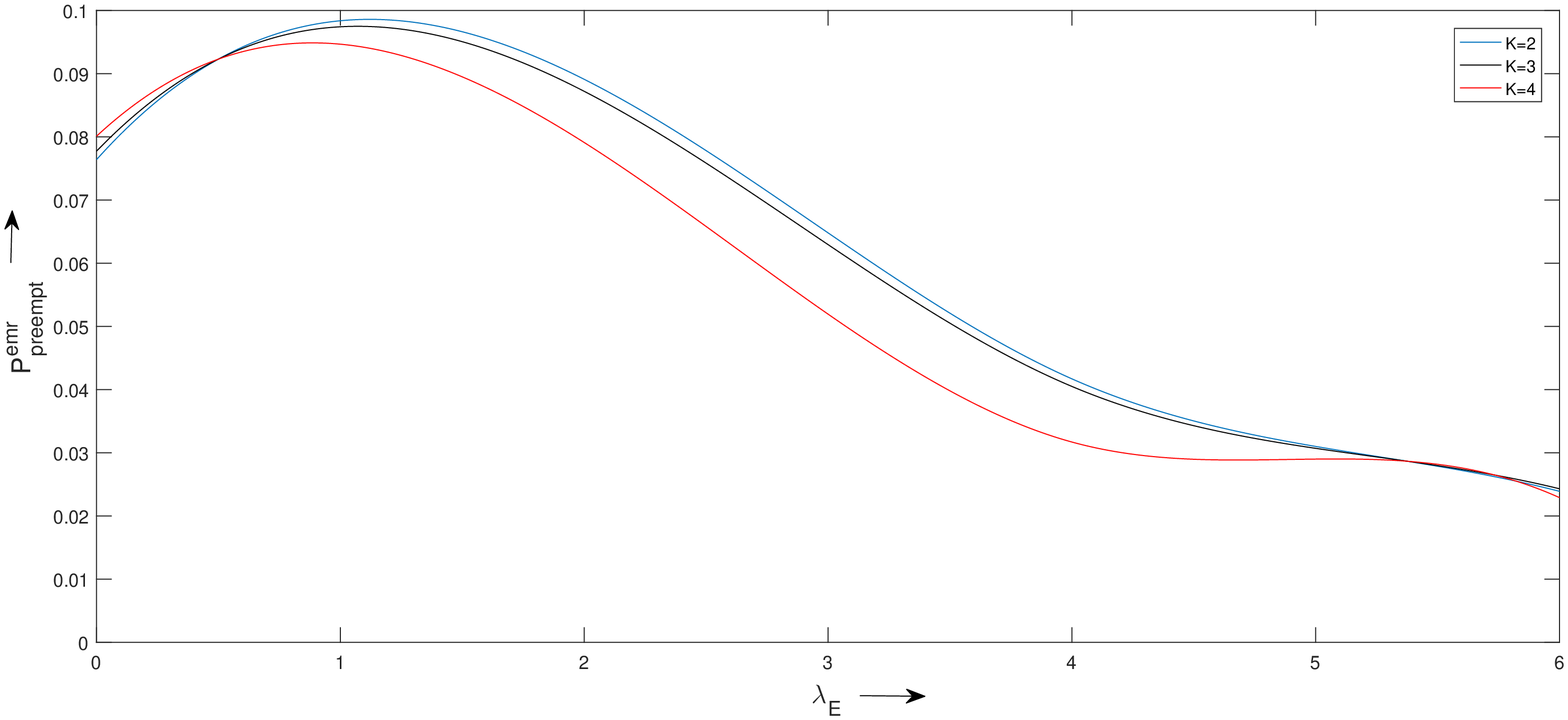}
		\label{fig:5a}}%
	\subfigure[$P_{preempt}^{emr}$ versus  $\mu_{\mathcal{E}}$]
	{\includegraphics[trim= 2.5cm 0.05cm 2cm 0.3cm, clip=true, height = 5.5cm,width = 0.5\textwidth]{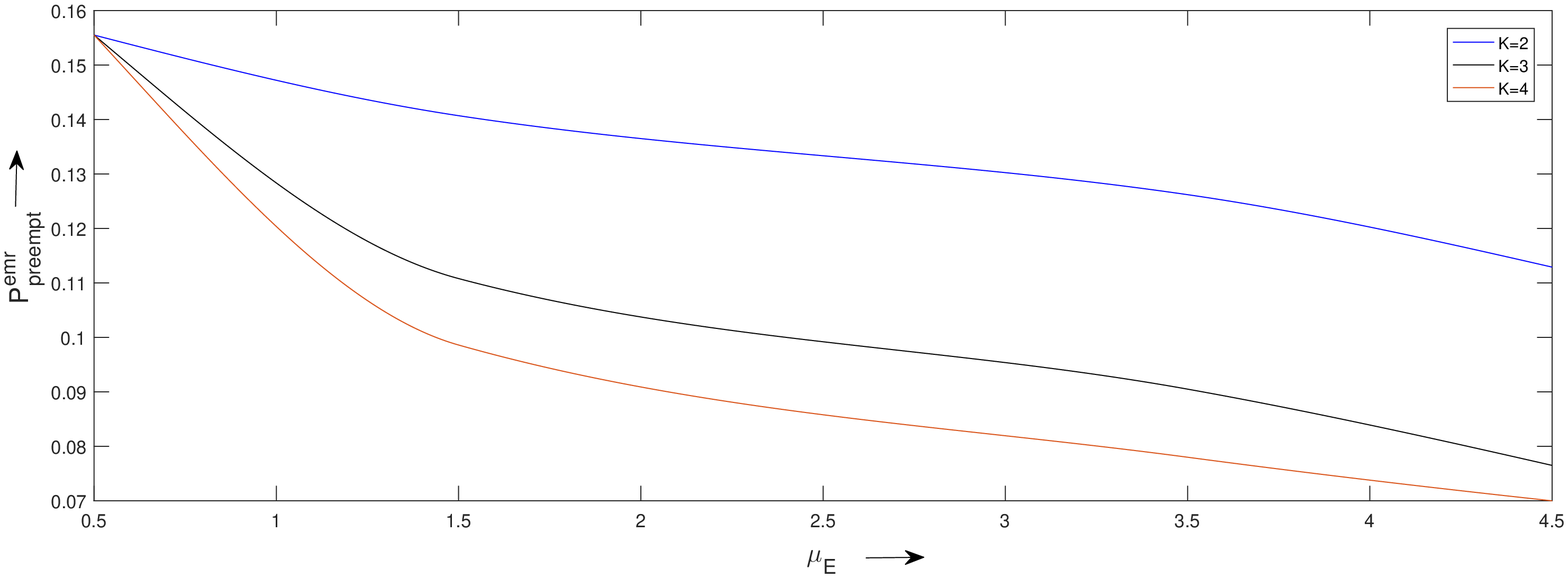}
		\label{fig:5b}}%
	\caption{Dependence of the preemption probability for emergency call $P_{preempt}^{emr}$   over arrival rate of an emergency call $\lambda_{\mathcal{E}}$ and service rate of an emergency call $\mu_{\mathcal{E}}$.  }
	\label{fig:Pb_Pd2}
\end{figure}

\section{Optimization Problem} \label{section6}
\begin{table}[htp]
	\centering
	\scalebox{1.1}{
	\begin{tabular}{|llllll|}
		\hline
		S=2 & $\mu_{\mathcal{N}}= 1$& $\mu_{\mathcal{H}}= 1$ & $\lambda_{\mathcal{N}}=1$ & $\lambda_{\mathcal{H}}=1$ & $\theta=1$\\
		\hline
	$K^*$&  $\lambda_{\mathcal{E}}^*$& $\mu_{\mathcal{E}}^*$ & $P_{\mathcal{E}}$ &$P_b^c$ & $P_{preempt}^{emr}$\\
	\hline
	1 & 1.525 & 4.7344 & 0.000160 & 0.000324 &0.000164\\
		\hline
		S=2 & $\mu_{\mathcal{N}}= 1$& $\mu_{\mathcal{H}}= 1.5$ & $\lambda_{\mathcal{N}}=1$ & $\lambda_{\mathcal{H}}=1$ &$\theta=1$\\
		\hline
	$K^*$&  $\lambda_{\mathcal{E}}^*$& $\mu_{\mathcal{E}}^*$ & $P_{\mathcal{E}}$ &$P_b^c$ & $P_{preempt}^{emr}$\\
		\hline
	1 & 1.525 & 4.774 & 0.000151 & 0.000235 &0.000159\\
		\hline
	S=2 & $\mu_{\mathcal{N}}= 1$& $\mu_{\mathcal{H}}= 2$ & $\lambda_{\mathcal{N}}=1$ & $\lambda_{\mathcal{H}}=1$ &$\theta=1$\\
		\hline
	$K^*$&  $\lambda_{\mathcal{E}}^*$& $\mu_{\mathcal{E}}^*$ & $P_{\mathcal{E}}$ &$P_b^c$ & $P_{preempt}^{emr}$\\
		\hline
	1 & 1.525 & 4.811 & 0.000146 & 0.000154 & 0.000120\\
		\hline
		S=2 & $\mu_{\mathcal{N}}= 1$& $\mu_{\mathcal{H}}= 2.5$ & $\lambda_{\mathcal{N}}=1$ & $\lambda_{\mathcal{H}}=1$ &$\theta=1$\\
		\hline
	$K^*$&  $\lambda_{\mathcal{E}}^*$& $\mu_{\mathcal{E}}^*$ & $P_{\mathcal{E}}$ &$P_b^c$ & $P_{preempt}^{emr}$\\
		\hline
	1 & 1.525 & 4.944 & 0.000124 & 0.000098 &0.000101\\
		\hline
		S=2 & $\mu_{\mathcal{N}}= 1$& $\mu_{\mathcal{H}}= 3$ & $\lambda_{\mathcal{N}}=1$ & $\lambda_{\mathcal{H}}=1$ &$\theta=1$\\
		\hline
	$K^*$&  $\lambda_{\mathcal{E}}^*$& $\mu_{\mathcal{E}}^*$ & $P_{\mathcal{E}}$ &$P_b^c$ & $P_{preempt}^{emr}$\\
		\hline
		1 & 1.525 & 4.994 & 0.000121 & 0.000097 &0.000100\\
		\hline
			\hline
	S=3 & $\mu_{\mathcal{N}}= 1$& $\mu_{\mathcal{H}}= 1$ & $\lambda_{\mathcal{N}}=1$ & $\lambda_{\mathcal{H}}=1$ &$\theta=1$\\
		\hline
	$K^*$&  $\lambda_{\mathcal{E}}^*$& $\mu_{\mathcal{E}}^*$ & $P_{\mathcal{E}}$ &$P_b^c$ & $P_{preempt}^{emr}$\\
	\hline
	2 & 1.755 & 6.66995 & 0.000028456 & 0.000030929 & 0.00004726\\
		\hline
				S=3 & $\mu_{\mathcal{N}}= 1$& $\mu_{\mathcal{H}}= 1$ & $\lambda_{\mathcal{N}}=1$ & $\lambda_{\mathcal{H}}=1$ &$\theta=1$\\
		\hline
	$K^*$&  $\lambda_{\mathcal{E}}^*$& $\mu_{\mathcal{E}}^*$ & $P_{\mathcal{E}}$ &$P_b^c$ & $P_{preempt}^{emr}$\\
	\hline
2 & 1.755 & 6.67755 & 0.000026044 & 0.000029999 & 0.00004694\\
		\hline
			S=3 & $\mu_{\mathcal{N}}= 1$& $\mu_{\mathcal{H}}= 1$ & $\lambda_{\mathcal{N}}=1$ & $\lambda_{\mathcal{H}}=1$ &$\theta=1$\\
		\hline
	$K^*$&  $\lambda_{\mathcal{E}}^*$& $\mu_{\mathcal{E}}^*$ & $P_{\mathcal{E}}$ &$P_b^c$ & $P_{preempt}^{emr}$\\
	\hline
2 & 1.755 & 6.67424 & 0.000024114 & 0.0000287414 & 0.000045422\\
		\hline
			S=3 & $\mu_{\mathcal{N}}= 1$& $\mu_{\mathcal{H}}= 1$ & $\lambda_{\mathcal{N}}=1$ & $\lambda_{\mathcal{H}}=1$ &$\theta=1$\\
		\hline
	$K^*$&  $\lambda_{\mathcal{E}}^*$& $\mu_{\mathcal{E}}^*$ & $P_{\mathcal{E}}$ &$P_b^c$ & $P_{preempt}^{emr}$\\
	\hline
2 & 1.755 & 6.65412 & 0.000022987 & 0.0000265415 & 0.000044111\\
		\hline
			S=3 & $\mu_{\mathcal{N}}= 1$& $\mu_{\mathcal{H}}= 1$ & $\lambda_{\mathcal{N}}=1$ & $\lambda_{\mathcal{H}}=1$ &$\theta=1$\\
		\hline
	$K^*$&  $\lambda_{\mathcal{E}}^*$& $\mu_{\mathcal{E}}^*$ & $P_{\mathcal{E}}$ &$P_b^c$ & $P_{preempt}^{emr}$\\
	\hline
2 & 1.755 & 6.66995 & 0.000028456 & 0.000030929 & 0.00004726\\
		\hline
			\hline
		S=4 & $\mu_{\mathcal{N}}= 1$& $\mu_{\mathcal{H}}= 1$ & $\lambda_{\mathcal{N}}=1$ & $\lambda_{\mathcal{H}}=1$ &$\theta=1$\\
		\hline
	$K^*$&  $\lambda_{\mathcal{E}}^*$& $\mu_{\mathcal{E}}^*$ & $P_{\mathcal{E}}$ &$P_b^c$ & $P_{preempt}^{emr}$\\
	\hline
2 & 2.25 & 6.85295 & 0.0001546 & 0.0002345 & 0.000081288\\
		\hline
			S=4 & $\mu_{\mathcal{N}}= 1$& $\mu_{\mathcal{H}}= 1.5$ & $\lambda_{\mathcal{N}}=1$ & $\lambda_{\mathcal{H}}=1$ &$\theta=1$\\
		\hline
	$K^*$&  $\lambda_{\mathcal{E}}^*$& $\mu_{\mathcal{E}}^*$ & $P_{\mathcal{E}}$ &$P_b^c$ & $P_{preempt}^{emr}$\\
	\hline
2 & 2.25 & 6.87745 & 0.0001224 & 0.0002289 & 0.0000800012\\
		\hline
			S=4 & $\mu_{\mathcal{N}}= 1$& $\mu_{\mathcal{H}}= 2$ & $\lambda_{\mathcal{N}}=1$ & $\lambda_{\mathcal{H}}=1$ &$\theta=1$\\
		\hline
	$K^*$&  $\lambda_{\mathcal{E}}^*$& $\mu_{\mathcal{E}}^*$ & $P_{\mathcal{E}}$ &$P_b^c$ & $P_{preempt}^{emr}$\\
	\hline
2 & 2.25 & 6.89455 & 0.00011124 & 0.00022241 & 0.0000800005\\
		\hline
			S=4 & $\mu_{\mathcal{N}}= 1$& $\mu_{\mathcal{H}}= 2.5$ & $\lambda_{\mathcal{N}}=1$ & $\lambda_{\mathcal{H}}=1$ &$\theta=1$\\
		\hline
	$K^*$&  $\lambda_{\mathcal{E}}^*$& $\mu_{\mathcal{E}}^*$ & $P_{\mathcal{E}}$ &$P_b^c$ & $P_{preempt}^{emr}$\\
	\hline
2 & 2.25 & 6.9041 & 0.00010012 & 0.00021442 & 0.0000799945\\
		\hline
			S=4 & $\mu_{\mathcal{N}}= 1$& $\mu_{\mathcal{H}}= 3$ & $\lambda_{\mathcal{N}}=1$ & $\lambda_{\mathcal{H}}=1$ &$\theta=1$\\
		\hline
	$K^*$&  $\lambda_{\mathcal{E}}^*$& $\mu_{\mathcal{E}}^*$ & $P_{\mathcal{E}}$ &$P_b^c$ & $P_{preempt}^{emr}$\\
	\hline
2 & 2.25 & 6.91455 & 0.000100044 & 0.00021141 & 0.0000788564\\
		\hline
		
	\end{tabular}}
	\caption{Optimal values of $\mu_{\mathcal{E}}^*$, $\lambda_{\mathcal{E}}^*$ and $K$ for different values of $S$ by applying NSGA-II method.}
	\label{tab:my_label1}
\end{table}

\noindent In the catastrophic scenario, the loss probabilities like blocking probability for emergency calls, blocking probability for new/handoff calls, and preemption probability for emergency calls are  considered as performance determining factors for  cellular networks. An increment in $P_{\mathcal{E}}$, $P_b^c$  as well as in $P_{preempt}^{emr}$ indicates unsatisfactory level of service for the customers.  Thus, it is required to find the optimal values of parameters in such a way that the loss probabilities should not exceed some pre-defined values. In Section \ref{section5}, a detailed analysis of loss probabilities with respect to several parameters has been provided. It can be observed from the results that $P_{\mathcal{E}}$, $P_b^c$  as well as in $P_{preempt}^{emr}$  are mostly affected by the arrival rate of emergency calls $\lambda_{\mathcal{E}}$, service rate of emergency calls $\mu_{\mathcal{E}}$ and the total number of backup channels $K.$ The service provider certainly cannot determine $\lambda_{\mathcal{E}}$, yet an approximated value of $\lambda_{\mathcal{E}}$ can be estimated  in order to keep sufficient backup channels to provide service. Thus, a non-trivial optimization problem is proposed with the decision variables $K$, $\mu_{\mathcal{E}}$ and $ \lambda_{\mathcal{E}}$ given as:\\
\begin{center}
$\begin{array}{lll}
&\textrm{min } &K\\
&\textrm{subject to},& P_{\mathcal{E}}( K,\lambda_{\mathcal{E}},\mu_{\mathcal{E}}) \leq \epsilon_1,\\
&&P_{b}^{c}(K,\lambda_{\mathcal{E}},\mu_{\mathcal{E}}) \leq \epsilon_2, \\
&&P_{preempt}^{emr}(K,\lambda_{\mathcal{E}},\mu_{\mathcal{E}}) \leq \epsilon_3, \\
&& K \leq S,\\
&& K,\lambda_{\mathcal{E}},\mu_{\mathcal{E}} \geq 0.
\end{array}$\\
\end{center}
Here, $\epsilon_1$, $\epsilon_2$  and $\epsilon_3$ are pre-defined values depending on the tolerance of the system for $P_{\mathcal{E}}$, $P_b^c$ and  $P_{preempt}^{emr}$, respectively.  Consider $\epsilon_1=\epsilon_2=\epsilon_3=10^{-3}$  for the numerical computation of the proposed optimization problem. 
The above defined  constraints are non-linear and highly complex.  Thus, a multi-objective evolutionary approach, non-dominated sorting genetic algorithm-II (NSGA-II) has been used to obtain the optimal solution. A modified version of non-dominated sorting genetic algorithm,  was proposed by Deb et al. \cite{deb2002fast} as NSGA-II. The detailed analysis of NSGA-II algorithm can be found in \cite{deb2002fast}. The main steps of NSGA-II are provided as follows. 
\begin{itemize}
    \item[1.] Initialize the population size $P$ based on the number of decision variables.
    \item[2.] \textbf{Non-dominated sorting:} The initialized population is sorted on the basis of non-domination. Each solution is assigned a fitness or rank equal to its non-domination level. The steps of sort algorithm are as follows:
    \begin{itemize}
        \item Initialize $S_p= \phi$. $S_p$ is set of all individuals that is being dominated by $p.$
        \item Initialize $n_p=0.$ This would be the number of individuals that dominate $p.$
        \item For each individual $q$ in $P$, if $p$ dominated $q$ then add $q$ to the set $S_p$ i.e., $S_p=S_p \cup \{q\} $ else if $q$ dominates $p$ then increment the domination counter for $p$ i.e., $n_p=n_p+1.$
        \item If $n_p=0$ i.e., no individuals dominate $p$ then $p$ belongs to the first front; Set rank of individual $p$ to 1. Update the first front i.e., $F_1 = F_1 \cup \{p\}.$
        This is carried out for all the individuals in main population $P.$ 
        \item Initialize the front counter to one, i.e., $i=1.$
        \item Define $Q=\phi$, the set for storing the individuals for $(i+1)^{th}$ front. Set $n_q=n_q-1$, decrements the domination count for individual $q$ in $S_p.$ Update the set $Q=Q \cup q.$
        \item Increment the front counter by one and set the next front as $F_i=Q.$ This step is carried out while the $i^{th}$ front is non empty.
    \end{itemize}
    \item[3.] \textbf{Crowding Distance:} All the individuals after non-dominated sort are assigned a crowding distance value. Crowding distance is assigned front wise and compared between two individuals.
    \item[4.] \textbf{Selection:} Once the individuals are sorted based on non-domination and with crowding distance assigned, the selection is carried out using a crowded-comparison operator.
    \item[5.] \textbf{Recombination and selection:} The offspring population is combined with the current generation population and selection is performed to set the individuals of the next generation. 
\end{itemize}

Tables \ref{tab:my_label1} and \ref{tab:my_label2} represents the optimal values of $k^*$, $\lambda_{\mathcal{E}}^*$, and $\mu_{\mathcal{E}}^*$ for different combinations of arrival and service rates and different values of $S.$  All results were obtained by MATLAB software, which were run on a computer with Intel Core i7-6700 3.40GHz CPU and 8 GB of RAM. The obtained results provide the value of optimal backup channels for various combinations of  total number of channels and arrival rates. The proposed optimization problem's sensitivity analysis is useful in estimating the number of backup channels in an emergency scenario. These findings are extremely beneficial in communication systems and cellular networks.

\begin{table}
	\centering
	\scalebox{1.1}{
	\begin{tabular}{|llllll|}
		\hline
		S=5 & $\mu_{\mathcal{N}}= 1$& $\mu_{\mathcal{H}}= 1$ & $\lambda_{\mathcal{N}}=1$ & $\lambda_{\mathcal{H}}=1$ &$\theta=1$\\
		\hline
	$K^*$&  $\lambda_{\mathcal{E}}^*$& $\mu_{\mathcal{E}}^*$ & $P_{\mathcal{E}}$ &$P_b^c$ & $P_{preempt}^{emr}$\\
	\hline
3 & 1.654 & 5.56205 & 0.000081836 & 0.000089307 & 0.00007471\\
		\hline
			S=5 & $\mu_{\mathcal{N}}= 1$& $\mu_{\mathcal{H}}= 1.5$ & $\lambda_{\mathcal{N}}=1$ & $\lambda_{\mathcal{H}}=1$ &$\theta=1$\\
		\hline
	$K^*$&  $\lambda_{\mathcal{E}}^*$& $\mu_{\mathcal{E}}^*$ & $P_{\mathcal{E}}$ &$P_b^c$ & $P_{preempt}^{emr}$\\
	\hline
3 & 1.654 & 5.5741 & 0.000080241 & 0.000088742 & 0.000073341\\
		\hline
			S=5 & $\mu_{\mathcal{N}}= 1$& $\mu_{\mathcal{H}}= 2$ & $\lambda_{\mathcal{N}}=1$ & $\lambda_{\mathcal{H}}=1$ &$\theta=1$\\
		\hline
	$K^*$&  $\lambda_{\mathcal{E}}^*$& $\mu_{\mathcal{E}}^*$ & $P_{\mathcal{E}}$ &$P_b^c$ & $P_{preempt}^{emr}$\\
	\hline
3 & 1.654 & 5.5787 & 0.000080240 & 0.000088701 & 0.000073339\\
		\hline
			S=5 & $\mu_{\mathcal{N}}= 1$& $\mu_{\mathcal{H}}= 2.5$ & $\lambda_{\mathcal{N}}=1$ & $\lambda_{\mathcal{H}}=1$ &$\theta=1$\\
		\hline
	$K^*$&  $\lambda_{\mathcal{E}}^*$& $\mu_{\mathcal{E}}^*$ & $P_{\mathcal{E}}$ &$P_b^c$ & $P_{preempt}^{emr}$\\
	\hline
3 & 1.654 & 5.5874 & 0.000080014 & 0.000087741 & 0.00007244\\
		\hline
			S=5 & $\mu_{\mathcal{N}}= 1$& $\mu_{\mathcal{H}}= 3$ & $\lambda_{\mathcal{N}}=1$ & $\lambda_{\mathcal{H}}=1$ &$\theta=1$\\
		\hline
	$K^*$&  $\lambda_{\mathcal{E}}^*$& $\mu_{\mathcal{E}}^*$ & $P_{\mathcal{E}}$ &$P_b^c$ & $P_{preempt}^{emr}$\\
	\hline
3 & 1.654 & 5.59874 & 0.0000879845 & 0.000086441 & 0.000070011\\
		\hline
			\hline
		S=6 & $\mu_{\mathcal{N}}= 1$& $\mu_{\mathcal{H}}= 1$ & $\lambda_{\mathcal{N}}=1$ & $\lambda_{\mathcal{H}}=1$ &$\theta=1$\\
		\hline
	$K^*$&  $\lambda_{\mathcal{E}}^*$& $\mu_{\mathcal{E}}^*$ & $P_{\mathcal{E}}$ &$P_b^c$ & $P_{preempt}^{emr}$\\
	\hline
4 & 1.5 & 5.74075 & 0.000047512 & 0.000057343 & 0.00009831\\
		\hline
			S=6 & $\mu_{\mathcal{N}}= 1$& $\mu_{\mathcal{H}}= 1.5$ & $\lambda_{\mathcal{N}}=1$ & $\lambda_{\mathcal{H}}=1$ &$\theta=1$\\
		\hline
	$K^*$&  $\lambda_{\mathcal{E}}^*$& $\mu_{\mathcal{E}}^*$ & $P_{\mathcal{E}}$ &$P_b^c$ & $P_{preempt}^{emr}$\\
	\hline
4 & 1.5 & 5.7221 & 0.000045442 & 0.000056134 & 0.000097745\\
		\hline
			S=6 & $\mu_{\mathcal{N}}= 1$& $\mu_{\mathcal{H}}= 2$ & $\lambda_{\mathcal{N}}=1$ & $\lambda_{\mathcal{H}}=1$ &$\theta=1$\\
		\hline
	$K^*$&  $\lambda_{\mathcal{E}}^*$& $\mu_{\mathcal{E}}^*$ & $P_{\mathcal{E}}$ &$P_b^c$ & $P_{preempt}^{emr}$\\
	\hline
4 & 1.5 & 5.7116 & 0.00004487 & 0.000055443 & 0.0000966412\\
		\hline
			S=6 & $\mu_{\mathcal{N}}= 1$& $\mu_{\mathcal{H}}= 2.5$ & $\lambda_{\mathcal{N}}=1$ & $\lambda_{\mathcal{H}}=1$ &$\theta=1$\\
		\hline
	$K^*$&  $\lambda_{\mathcal{E}}^*$& $\mu_{\mathcal{E}}^*$ & $P_{\mathcal{E}}$ &$P_b^c$ & $P_{preempt}^{emr}$\\
	\hline
4 & 1.5 & 5.70098 & 0.000042114 & 0.00005474 & 0.0000955414\\
		\hline
			S=6 & $\mu_{\mathcal{N}}= 1$& $\mu_{\mathcal{H}}= 3$ & $\lambda_{\mathcal{N}}=1$ & $\lambda_{\mathcal{H}}=1$ &$\theta=1$\\
		\hline
	$K^*$&  $\lambda_{\mathcal{E}}^*$& $\mu_{\mathcal{E}}^*$ & $P_{\mathcal{E}}$ &$P_b^c$ & $P_{preempt}^{emr}$\\
	\hline
4 & 1.5 & 5.70014 & 0.000042001 & 0.00005470 & 0.000095540\\
		\hline
	\end{tabular}}
	\caption{Optimal values of $\mu_{\mathcal{E}}^*$, $\lambda_{\mathcal{E}}^*$ and $K$ for different values of $S$ by applying NSGA-II method.}
	\label{tab:my_label2}
\end{table}

\section{Conclusions} \label{section7}
Queuing models with catastrophic occurrences are a driving force in communications and cellular networks in cutting-edge wireless technology. In such types of catastrophic queueing models, traffic of different classes, e.g., video, voice, images, data, etc., are assigned different categories of importance,  and consequently their services are effectuated in accordance with an appropriate priority policy.  In cellular networks, the kinds of systems where a higher priority traffic has an advantage in access to service compared to less important ones, are explored through priority policies. Therefore, in this study  a $\textrm{\it MMAP[k]/PH[k]/S}$  catastrophic  queueing model with preemptive repeat priority policy and \textrm{$P\!H$} distributed retrial times is investigated. Due to the brief span of inter-retrial times in comparison to service times, a more generalized approach, \textrm{$P\!H$} distributed retrial
times is used so that the performance of the system is not over or under estimated. The proposed model deals with three types of incoming calls, handoff call, new call and emergency call. Before the occurrence of a catastrophe, handoff call are provided preemptive priority over new call. Whereas, in the catastrophic scenario, emergency calls are provided preemptive priority over handoff and new calls. When a calamity strikes, backup channels are used to establish communication in the affected area. The underlying process of the presented system is modelled by \textrm{$A\!Q\!T\!M\!C$}. Ergodicity conditions of the underlying Markov chain are obtained by proving that the Markov chain belongs to the class of \textrm{$A\!Q\!T\!M\!C$}. A new algorithm is developed for approximate computation of the stationary distribution. Further, the expressions of various performance measures have been derived for the numerical illustration. Due to the  consideration of the preemptive  priority policy, the blocking probability for emergency calls decreases and simultaneously the frequent termination of services for handoff and new calls increases the probability of preemption. Thus, an optimization problem to obtain optimal value of total number of  backup channels  has been formulated and dealt by employing NSGA-II approach.

\end{document}